\documentclass{sig-alternate-05-2015}
\pdfoutput=1
\usepackage{graphicx}
\usepackage{balance}  
\usepackage{amsmath}
\usepackage{amsfonts}
\usepackage{amssymb}
\usepackage{subfig}
\usepackage{caption}
\usepackage{booktabs}
\usepackage{multirow}
\usepackage{epstopdf}
\usepackage{algorithmicx, algpseudocode}
\usepackage[linesnumbered,ruled,vlined]{algorithm2e}
\usepackage{url}
\usepackage[font={small,it}]{caption}
\usepackage{color}

\newtheorem{theorem}{Theorem}

\newtheorem{Lemma}{Lemma}
\newtheorem{Definition}{Definition}

\newcommand\E{\mathbb{E}}

\newcommand\R{\mathcal{R}}
\def\RIS{\textsf{RIS}}
\def\TRIS{\textsf{TRIS}}

\def\SI{\textsf{SI}}
\def\ISI{\textsf{ISI}}
\def\IC{\textsf{IC}}
\def\LT{\textsf{LT}}
\def\SISI{\textsf{SISI}}
\def\RR{\textsf{RR}}
\def\TRR{\textsf{RR}}
\def\src{\textsf{src}}
\def\NETSLEUTH{\textsf{NETSLEUTH}}

\begin{document}

\makeatletter
\def\@copyrightspace{\relax}
\makeatother
%

%



%

\title{Multiple Infection Sources Identification with\\ Provable Guarantees}
%
%
%
%
%

\numberofauthors{3} 
%
\author{
	\alignauthor
	H. T. Nguyen, P. Ghosh\\
	\affaddr{CS Department}\\
	\affaddr{Virginia Commonwealth Univ.}\\
	\affaddr{Richmond, VA 23284, USA}\\
	\affaddr{\{hungnt, pghosh\}@vcu.edu} 
	\alignauthor
	M. L. Mayo\\
	\affaddr{US Army Engineer RD Center}\\
	\affaddr{3909 Halls Ferry Road, }\\
	\affaddr{Vicksburg, MS 39180, USA}\\
	\affaddr{Michael.L.Mayo@usace.army.mil}
	\alignauthor
	T. N. Dinh\\
	\affaddr{CS Department}\\
	\affaddr{Virginia Commonwealth Univ.}\\
	\affaddr{Richmond, VA 23284, USA}\\
	\affaddr{tndinh@vcu.edu} 
}

\maketitle
\begin{abstract}
Given an aftermath of a cascade in the network, i.e. a set $V_I$ of ``infected'' nodes after an epidemic outbreak or a propagation of rumors/worms/viruses, how can we infer the sources of the cascade? Answering this challenging question is critical for computer forensic, vulnerability analysis, and risk management.  Despite recent interest towards this problem, most of existing works  focus only on single source detection or simple network topologies, e.g. trees or grids. 

In this paper, we propose a new approach to identify infection sources by searching for a seed set $S$ that   minimizes the \emph{symmetric difference} between the cascade from $S$ and $V_I$, the given set of infected nodes.  Our major result is an approximation  algorithm, called  \SISI{}, to identify infection sources \emph{without the prior knowledge on the number of source nodes}. \SISI{}, to our best knowledge, is the first algorithm with \emph{provable guarantee} for the problem in general graphs. It returns a $\frac{2}{(1-\epsilon)^2}\Delta$-approximate solution with high probability, where  $\Delta$ denotes the maximum number of nodes in $V_I$ that may infect a single node in the network. 
Our experiments on real-world networks show the superiority of our approach and \SISI{} in detecting true source(s), boosting the F1-measure from few percents, for the state-of-the-art \NETSLEUTH{}, to approximately 50\%.
\end{abstract}

%
%


%
%

%
%


\renewcommand{\arraystretch}{1.2}
\setlength{\intextsep}{0.5\baselineskip}
\captionsetup{labelfont=bf}

\setlength\tabcolsep{3pt}

\section{Introduction}
\label{sec:intro}

The explosion of online social networks with billion of users such as Facebook or Twitter have fundamentally changed the landscapes of information sharing, nowadays.  Unfortunately, the same channels can be exploited to spread rumors and misinformation that cause devastating effects such as widespread panic in the general public \cite{Swine09}, diplomatic tensions \cite{Fox13}, and witch hunts towards innocent people \cite{Boston13}. 

Given a snapshot of the network with a set $V_I$ of \textit{infected nodes} who posted the rumors, identifying the set of nodes who initially spread the rumors is a challenging, yet important question, whether for forensic use or insights to prevent future epidemics. Other applications of infection source detection can be found in finding first computing devices that get infected with a virus or source(s) of contamination in water networks.

Despite recent interest towards this problem, termed Infection Sources Identifications (\ISI{}), most of existing works either limit to single source detection \cite{Luo13, Lokhov14} or simple network topologies, e.g. trees or grids, with ad hoc extensions to general graphs \cite{Lappas10,Shah11,Shah12,Luo13}. A recent work in \cite{Prakash12} provides an MDL-based method, called \NETSLEUTH{}, to detect both the number of infection sources and the sources themselves. However the proposed heuristics seems to only work well for grid networks and cannot detect any true infection source. Thus there is lack of a rigorous and accurate method to detect multiple infection sources in general graphs.

\begin{figure}[t]
	\centering
	\includegraphics[width=0.35\textwidth]{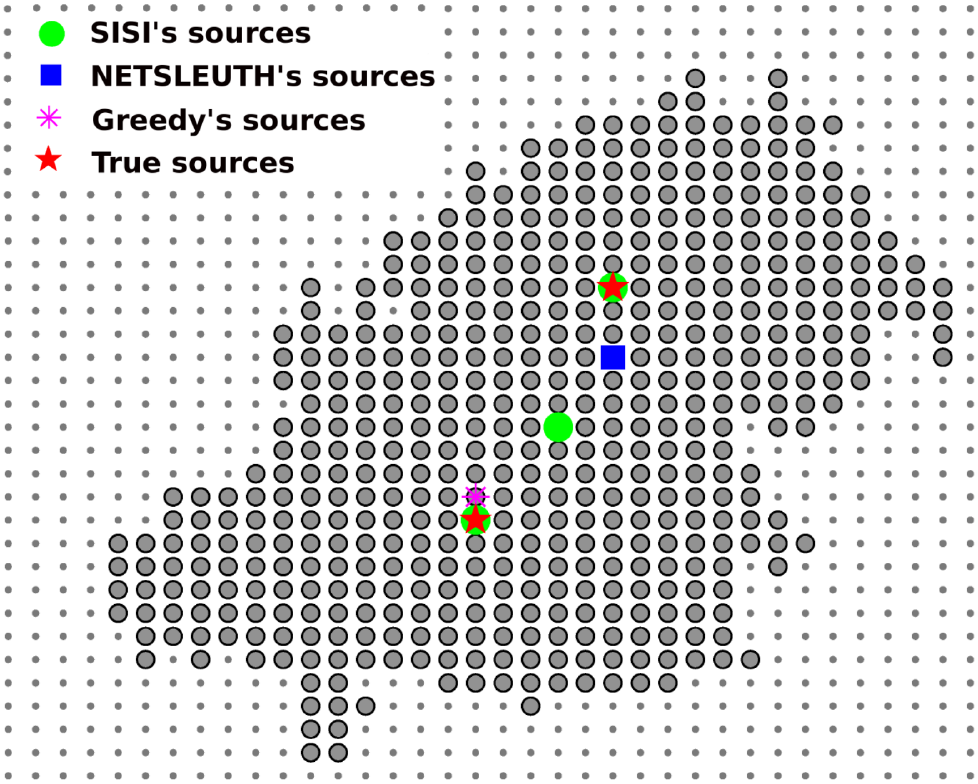}
	\caption{Infection sources detection on a $60\times 60$ grid graph}
	\label{fig:grid}
	\vspace{-0.2in}
\end{figure} 
In this paper, we present a new approach to identify \emph{multiple} infection sources that looks into both infected and uninfected nodes. This contradicts to existing methods \cite{Lappas10, Prakash12}  which limit the attention to the subgraph induced by the infected nodes. Given a snapshot of network $G=(V, E)$ and a set of infected nodes $V_I$, we identify the sources by searching for a set $\hat S$ that minimize the \emph{symmetric difference} between the cascade from $S$ and $V_I$. While our objective, the symmetric difference, is similar to the one used in \textsf{$k$-effector} \cite{Lappas10}, our novel formulation does not require the knowledge of the number of infection sources $k$. In deriving optimization method for this new approach, we face  \emph{strong challenges} in developing efficient solution:
\begin{itemize}
\vspace{-0.02in}
\item The \emph{exponential} number, up to $2^{\theta(n)}$ for large $V_I$, of possible solutions. This makes the exhaustive search for the case of single source \cite{Luo13, Lokhov14, Farajtabar15} intractable. 
\vspace{-0.02in}
\item The \emph{non-submodular} objective. Thus, it is inefficient to solve the problem through simple greedy  methods.
\vspace{-0.02in}
\item The \emph{stochastic} nature of the infection process requires considering exponentially many possible cascades.
\vspace{-0.05in}
\end{itemize}

To tackle \ISI{}, we propose  \SISI{}, an algorithm that can accurately detect infection sources. We employ in \SISI{} two novel techniques: a  \textit{Truncated Reverse Infection Sampling} (\TRIS) method to generate random reachability \TRR{} sets that encode the infection landscape and a primal-dual algorithm for the \textit{Submodular-cost Covering} \cite{Koufogiannakis13}. \SISI{}, to our best knowledge, is the \emph{first algorithm with provable guarantee} for multiple infection sources detection in general graphs. It returns an $\frac{2}{(1-\epsilon)^2}\Delta$-approximate solution with a high probability, where  $\Delta$ denotes the maximum number of nodes in $V_I$ that may infect the same node in the network.
Experiments on real-world networks show the huge leap of \SISI{} in detecting true infection sources, boosting the true source discovery rates from merely few percents, for the state-of-the-art \NETSLEUTH{}, to more than 70\%. Thus \SISI{} has both high empirical performance and theoretical guarantees. 

The advantages of \SISI{} over other methods are illustrated through a cascade on a $60\times 60$ grid in Fig.~\ref{fig:grid}. \SISI{} is the only one which can \emph{detect the true infection sources}. To avoid false negative, which is more serious than false positive, \SISI{} often output slightly more infection sources than other methods (\SISI: 3, \NETSLEUTH{}: 1, \textsf{Greedy}:1, Ground-truth: 2). However, it maintains a reasonable F1-score of over 50\%.

We summarize our contributions as follows
\vspace{-0.03in}
\begin{itemize}
  \item We propose a new approach to identify multiple infection sources through minimizing the symmetric difference between the cascade of the suspected source nodes $S$ with the infected nodes $V_I$ without knowing \textit{the number of sources a priori}. Our experiments show that methods following this approach including our algorithm \SISI{} and the greedy algorithm outperform the other approaches in terms of detecting true sources.
  \vspace{-0.03in}
  \item To our best knowledge, we propose the first approximation algorithm, termed \SISI{}, for detecting multiple infection sources in general graphs and our algorithm does not require the knowledge on the number of infection sources. Given an approximation error $\epsilon >0$, we  provide rigorous analysis on sample complexity, deriving the necessary number of samples to guarantee a multiplicative error $(1 \pm \epsilon)$ on the objective estimation.
  \vspace{-0.03in}
  \item Extensive experiments on real-world networks shows the superiority of \SISI{}  over other approaches in detecting the exact sources under both \SI{} and IC models. The relax version of \SISI{} is also faster than \NETSLEUTH{} while still retaining high-quality solutions.
\end{itemize}
\vspace{-0.03in}

\textbf{Related works.} Infection Source Identification (\ISI) under different names has recently emerged and attracted quite a number of researchers in multiple disciplines with diverse techniques. There are two main streams of works and methods that can be listed: 1) exact algorithms on tree graphs \cite{Shah11,Shah12,Luo13,Lappas10,Dong13}, 2) \textit{ad hoc} heuristics approaches without any guarantee for general graphs \cite{Prakash12,Luo13,Lokhov14}.

In the first stream, Shah and Zaman in \cite{Shah11} established the notion of rumor-centrality which is an Maximum Likelihood estimator on regular trees under the \SI{} model. They proposed an optimal algorithm to identify the single source of an epidemic. In \cite{Shah12}, the same authors improved the previous results by deriving the exact expression for the probability of correct detection. Later Luo et al. \cite{Luo13} based on approximations of the infection sequences count to develop an algorithm that can find at most two sources in a geometric tree. Since solely targeting trees, all these methods are unable of solving \ISI{} problem on general graphs.

Lappas et al. \cite{Lappas10} formulated \ISI{} problem under the name of \textsf{$k$-effector} and introduced the minimization of the symmetric difference between the observed infection and the resulting cascade if starting from a candidate source set.  While the formulation is novel, their solution is, unfortunately, limited to tree graphs and require the knowledge of the number of infection sources. The extension for general graphs by approximating a graph by a tree does not work well either as we show later in the experiment section.

Prakash et al. \cite{Prakash12} resort to heuristic approach to find multiple sources in general graphs and propose \NETSLEUTH{} which relies on the two-part code Minimum Description Length. They show that \NETSLEUTH{} is able to detect both the sources and how many of them. However, besides no guarantee on solution quality, we show in our experiments that \NETSLEUTH{} performs poorly on a simple grid graph with large overlapping region of cascades from two source nodes. Luo et al. \cite{Luo13} also derived an estimator to find multiple sources given that the maximum number of sources is provided. Yet similar to \cite{Lappas10}, their estimator depends on the approximation of a general graph to tree and also requires the maximum number of sources.

There are also other works on related areas: \cite{Karamchandani13} studies the rumor-centrality estimator on trees under an additional constraint that the status (infected or not) of a node is revealed with probability $p \leq 1$. In case of $p = 1$, the estimator is able to reproduce the previous results and with large enough $p < 1$, it achieves performance within $\epsilon$ the optimal. Under a different model, Chen et al. \cite{Chen14} study the problem of detecting multiple information sources in networks under the Susceptible-Infected-Recovered (SIR) model. They propose an estimator for regular trees that can detect sources within a constant distance to the real ones with high probability and investigate a heuristic algorithm for general cases. In another study \cite{Lokhov14}, Lokhov et al. take the dynamic message-passing approach under SIR model and introduce an inference algorithm which is shown to admit good improvement.

Influence maximization problem \cite{Kempe03} that find $k$ nodes to maximize the expected influence is one of the most extensively studied problem. The latest references on the problem can be found in \cite{Tang15} and the references therein.

\textbf{Organization.} We present our model and problem formulation in Section \ref{model}. The hardness and inapproximability results are provided in Section \ref{sec:hard}. The main algorithm \SISI{} is proposed in Section \ref{sec:submodular} while its performance guarantees and running time are analyzed in Section \ref{sec:approx}. We provide comparison on empirical performance of our algorithms and other approaches in Section \ref{sec:exp}. Conclusions and extensions for other settings are discussed in Section \ref{sec:conclusion}.
\vspace{-0.1in}

\section{Models and Problem definition}
\label{model}

We represent the network in which the infection spreads as a directed graph $G=(V, E)$ where $V$ is the set of $n$ nodes, e.g., computers in a computer network, and $E$ is the set of $m$ directed edges, e.g., connections between the computers. In addition, we are given a subset $V_I \subseteq V$ of observed infected nodes and the remaining nodes are assumed to be not infected and denoted by $\bar{V_{I}} = V\backslash V_{I}$.

\begin{table}[hbt]\small
	\centering
	\caption{Table of Notations}
	\vspace{-0.15in}
	\begin{tabular}{|p{1.9cm}|p{6.2cm}|}
		\addlinespace
		\toprule
		\bf Notation  &  \quad \quad \quad \bf Description \\
		\midrule 
		$n, m$ & \#nodes, \#edges of graph $G=(V, E)$.\\
		\hline
		$V_I, \bar V_I$ & Set of infected and uninfected nodes.\\
		\hline
		$\beta, k$ & Infection probability and $k = |V_I|$.\\
		\hline
		$V(S,\mathcal{M})$ & An infection cascade from $S$ under model $\mathcal{M}$.\\
		\hline
		$D(S,\mathcal{M},V_I)$ & Symmetric different on a graph realization.\\
		\hline
		$\E[D(S,\mathcal{M},V_I)]$ & The expectation of $D(S,\mathcal{M},V_I)$ over all realizations.\\
		\hline
		$\hat S$ & The returned source set of \SISI{}.\\
		\hline
		$OPT, S^*$ & The optimal value of $\E[D(S,\tau)]$ and an optimal solution set which achieves the optimal value.\\
		\hline
		$R_j, \src(R_j)$ & A random \TRR{} set and its source node $\src(R_j)$.\\
		\hline
		$\Delta$ & Maximum size of an \TRR{} set ($\Delta \leq V_I$).\\
		\hline
		$c, M$ & $c = 2(e-2) \approx \sqrt{2}$, $M = 2^k + 1$.\\
		\hline
		$ \Lambda$ & $\Lambda = (1 + \epsilon)2c(\ln \frac{2}{\delta} + k\ln 2 + 1) \frac{1}{\epsilon^2}$.\\
		\bottomrule
	\end{tabular}%
	\label{tab:syms}%
	\vspace{-0.1in}
\end{table}%
\subsection{Infection Model}
We focus on the popular \textit{Susceptible-Infected (\SI{})}  model.

\textbf{Susceptible-Infected (\SI) model.} In this infection model, each node in the network is in one of two states: 1) Susceptible (S) (not yet infected) and 2) Infected (I) (infected and capable of spreading the disease/rumor). Once infected, the node starts spreading to its neighbors through their connections. While the initial model were proposed for a complete graph topology \cite{Anderson92}, the model can be extended for arbitrary graph $G=(V, E)$. We assume that the infection spreads in discrete time steps. At time $t=0$, a subset of nodes, called the infection sources, are infected and the rest is uninfected. Once a node $u$ gets infected at time $t$, it will continuously try to infect its uninfected neighbor $v$ and succeed with probability $0 <\beta \leq 1$ from step $t+1$ onwards until successful. The single parameter $\beta$ indicates how contagious the infection is and thus the higher, the faster it contaminates the network. 

\textbf{Other cascade model.}   In principle, our formulation and proposed method will work for most progressive diffusion models in which once a node becomes infected, it stays infected. These include the two popular models \textit{Independent Cascade \IC{}} and \textit{Linear Threshold (\LT{})} models \cite{Kempe03}. Other non-progressive models can be first converted to a progressive ones as outlined in \cite{Chen13}. 

For simplicity, we present our method for the \SI{} model and discuss the extension to the \IC{} and \LT{} models through changing the sampling method in Subsection \ref{subsec:ris}.

\textbf{Learning model parameters.}  Learning propagation model parameters is an important topic and has received a great amount of interest \cite{Kimura09,Tang09,Goyal10,Liu10, Kutzkov13}. Our approaches can rely on  these learning methods to extract influence cascade model parameters from real datasets, e.g., action logs, connection networks.

\vspace{-0.05in}

\subsection{Problem Formulation}
Intuitively, given an infection model, denoted by $\mathcal M$, the goal of infection source identification is to identify a set of source nodes $S$ (unknown size) so that the resulting cascade originated from nodes in $S$, within a duration $\tau>0$, matches $V_I$ as closely as possible.

To formalize the problem, we define a cascade $V(S, \mathcal M)$ as the set of infected nodes if we select nodes in $S$ as the sources (initially infected) under infection model $\mathcal{M}$. Thus, the objective function which characterizes the aforementioned criteria, termed \textit{symmetric difference}, is defined as follows,
\begin{align}
	\label{eq:diff}
	D(S,\mathcal{M},V_I) = |V_{I}\backslash V(S,\mathcal{M})| + |\bar{V_{I}} \cap V(S,\mathcal{M})|
\end{align}
\captionsetup{labelfont=bf}
\begin{figure}[!ht] \centering
	\includegraphics[width=0.55\linewidth]{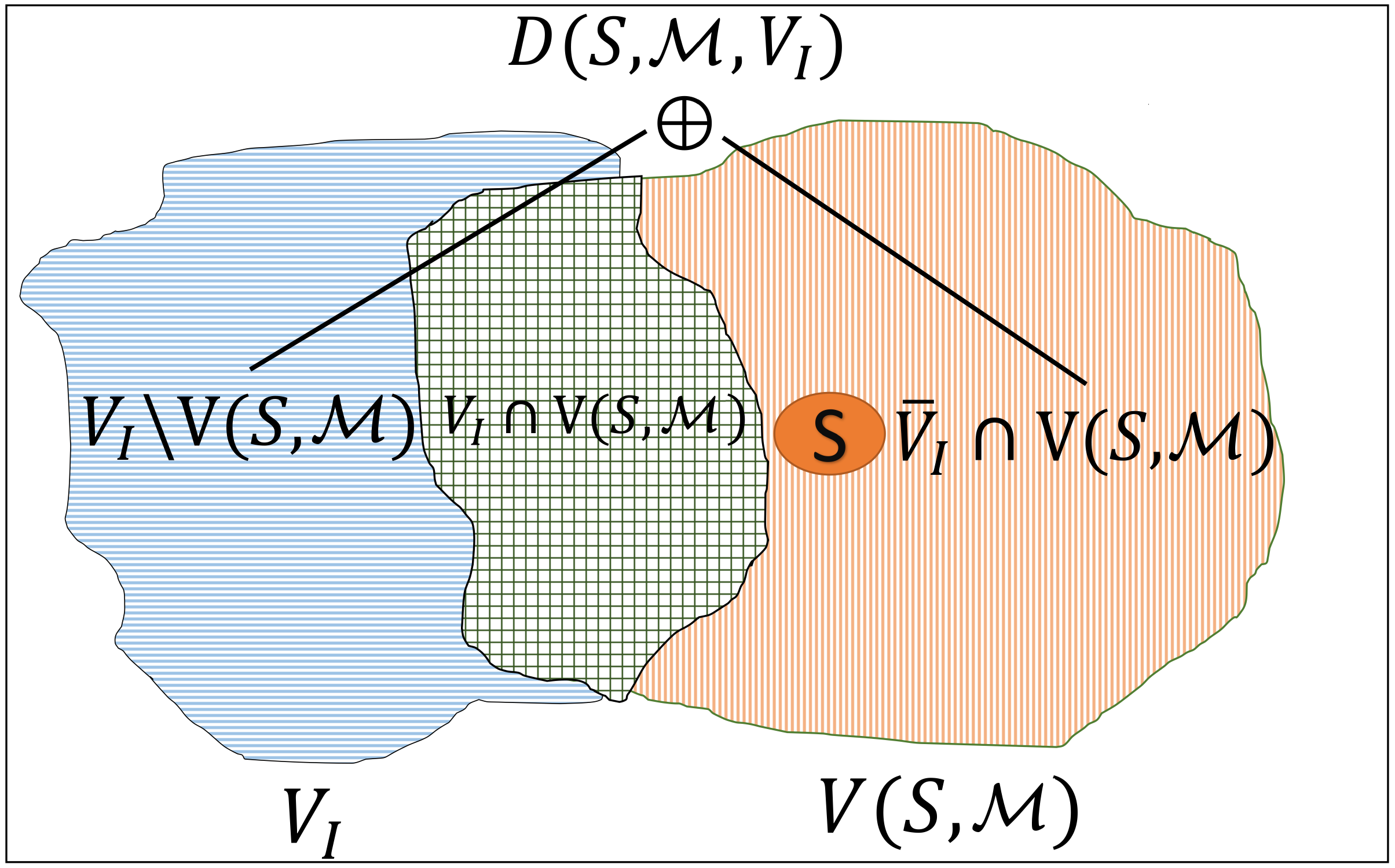}
	\caption{Illustration of symmetric difference.}
	\label{fig:formula}
\end{figure}

In Eq.~\ref{eq:diff}, the first term $|V_{I} \backslash V(S,\mathcal{M})|$ counts the number of nodes in $V_I$ that are not infected by the propagation spreading from $S$ within a duration $\tau$ and the second term indicates the number of nodes that are ``mistakenly'' infected during the same time interval (illustrated in Fig.~\ref{fig:formula}).  Together, the sum measures the similarity between the observed cascade $V_I$ and the cascade causes by the suspected nodes $S$. 

Due to the stochastic nature of the cascade, there are exponentially many possible cascades  for a given  set of source nodes $S$. Here cascade is used to refer to the set of infected nodes within $\tau$ steps. To account for this, we  aggregate the symmetric difference over the probabilistic space of the possible cascades spreading from $S$.  Denote by $\Pr[V(S,\mathcal{M})]$, the probability of receiving a particular cascade $V(S, \mathcal M)$ within $t=\tau$ time steps. We compute the expected symmetric difference as follows,
\begin{align}
	&\mathbb{E}[D (S,\mathcal{M},V_I)] = \sum_{\text{possible } V(S,\mathcal{M})}D(S,\mathcal{M},V_I) \Pr[V(S,\mathcal{M})] \nonumber \\
	&=\sum_{\text{possible } V(S,\mathcal{M})} (|V_{I}\backslash V(S,\mathcal{M})| + |\bar{V_{I}} \cap V(S,\mathcal{M})|)\Pr[V(S,\mathcal{M})] \nonumber \\
	&= \sum_{u \in V_{I}} \Pr[u \text{ not infected by } S] + \sum_{v \notin V_{I}}\Pr[v \text{ infected by } S]
	\label{eq:diff_2}
\end{align}
\noindent In the last equation, the `infected' and `not infected' probabilities are w.r.t. a random cascade from $S$ within $\tau$ steps. 

We now state the problem of identifying the infection sources as follows.
\begin{Definition}[Infection Sources Identification]
	\label{def:prob_1}
	Given a graph $G=(V,E)$, infection model $\mathcal{M}$ (e.g., $\beta$ for \SI{} model), observation set $V_I$ of infected nodes, and the duration of the cascade $\tau$ (could be infinity), the Infection Sources Identification (\ISI{}) problem asks for a set $\hat S$ of nodes such that,
	\begin{align}
		\label{eq:obj}
		\hat S = \arg\min_{S \subseteq V_{I}} \mathbb{E}[D(S,\mathcal{M},V_I)]
	\end{align}
\end{Definition}
While this formulation is similar to \cite{Lappas10}, we do not require knowledge on the number of infection sources.
\vspace{-0.05in}

\subsection{Hardness and Inapproximability}
\label{sec:hard}
This subsection shows the NP-hardness and inapproximability results of the \ISI{} problem. From Def.~\ref{def:prob_1}, there are two major difficulties in finding the sources: 1) first, by a similar argument to that of the influence maximization problem in \cite{Kempe03}, the objective function is \#P-hard to compute exactly; 2) second, the objective is non-submodular, i.e., there are no easy greedy approaches to obtain approximation algorithms. In fact, we show a stronger inapproximability result in the below theorem.

\begin{theorem}
	\label{theo:inapprox}
    \ISI{} cannot be approximated within a factor $O(2^{\log^{1-\epsilon}n})$ for any $\epsilon > 0$, where $n = |V|$, unless NP $\subseteq$ DTIME($n^{polylog(n)}$).
\end{theorem}
\begin{proof}
To prove Theo.~\ref{theo:inapprox}, we construct a gap-preserving polynomial-time reduction which reduces any instance of the Red-Blue Set Cover problem \cite{Carr00} to an instance of \ISI{}. The Red-Blue Set Cover problem is defined as follows: an instance of Red-Blue Set Cover problem consists of two disjoint sets: $R = \{r_1,...,r_p\}$ of red elements, $B = \{b_1,...,b_q\}$ of blue elements, and a family $T \subseteq 2^{R\cup B}$ of $n (n \geq p, n \geq q)$ subsets of $R\cup B$. The problem asks a subfamily $C^* \subseteq T$ of subsets that covers all the blue elements but minimum number of reds,
\begin{align}
	C^* = \arg\min_{C \subseteq T} \{|R\cap (\cup_{i=1}^{|C|} T_{i})|\}
\end{align}

Our polynomial reduction ensures that if the \ISI{} instance has an $O(2^{\log^{1-\epsilon}n})$-approximate solution $S$, then there must be a corresponding $O(2^{\log^{1-\epsilon}n})$-approximate solution of the Red-Blue Set Cover polynomially induced from $S$. The reduction is grounded on the observation that any solution of the Red-Blue Set Cover costs at most $p$ - the number of red elements. Then, based on the result in \cite{Carr00} that the Red-Blue Set Cover cannot be approximated within a factor $O(2^{\log^{1-\epsilon}N})$ where $N = n^4$ for any $\epsilon > 0$ unless NP $\subseteq$ DTIME($N^{polylog(N)}$), we obtain the Theorem~\ref{theo:inapprox}.

We will give a polynomial reduction from an instance of the Red-Blue Set Cover to an \ISI{} instance with $\beta = 1$ and $\tau = 1$ such that,
\begin{itemize}
	\item[(\textbf{1})] The optimal solution of the \ISI{} instance polynomially infers the optimal solution for the instance of Red-Blue Set Cover.
	\item[(\textbf{2})] If we obtain an $O(2^{\log^{1-\epsilon}n})$-approximate solution for \ISI{}, we will also have an $O(2^{\log^{1-\epsilon}n})$-approximate solution for the Red-Blue Set Cover instance.
\end{itemize}
These two conditions are sufficient to conclude that we cannot approximate the optimal solution of \ISI{} within a factor $O(2^{\log^{1-\epsilon}n})$ unless we can do that for Red-Blue Set Cover. Thus, the Theorem~\ref{theo:inapprox} follows. We will present the reduction and then prove the satisfaction of each condition.

\captionsetup{labelfont=bf}
\begin{figure}[!ht]
	\includegraphics[width=\linewidth]{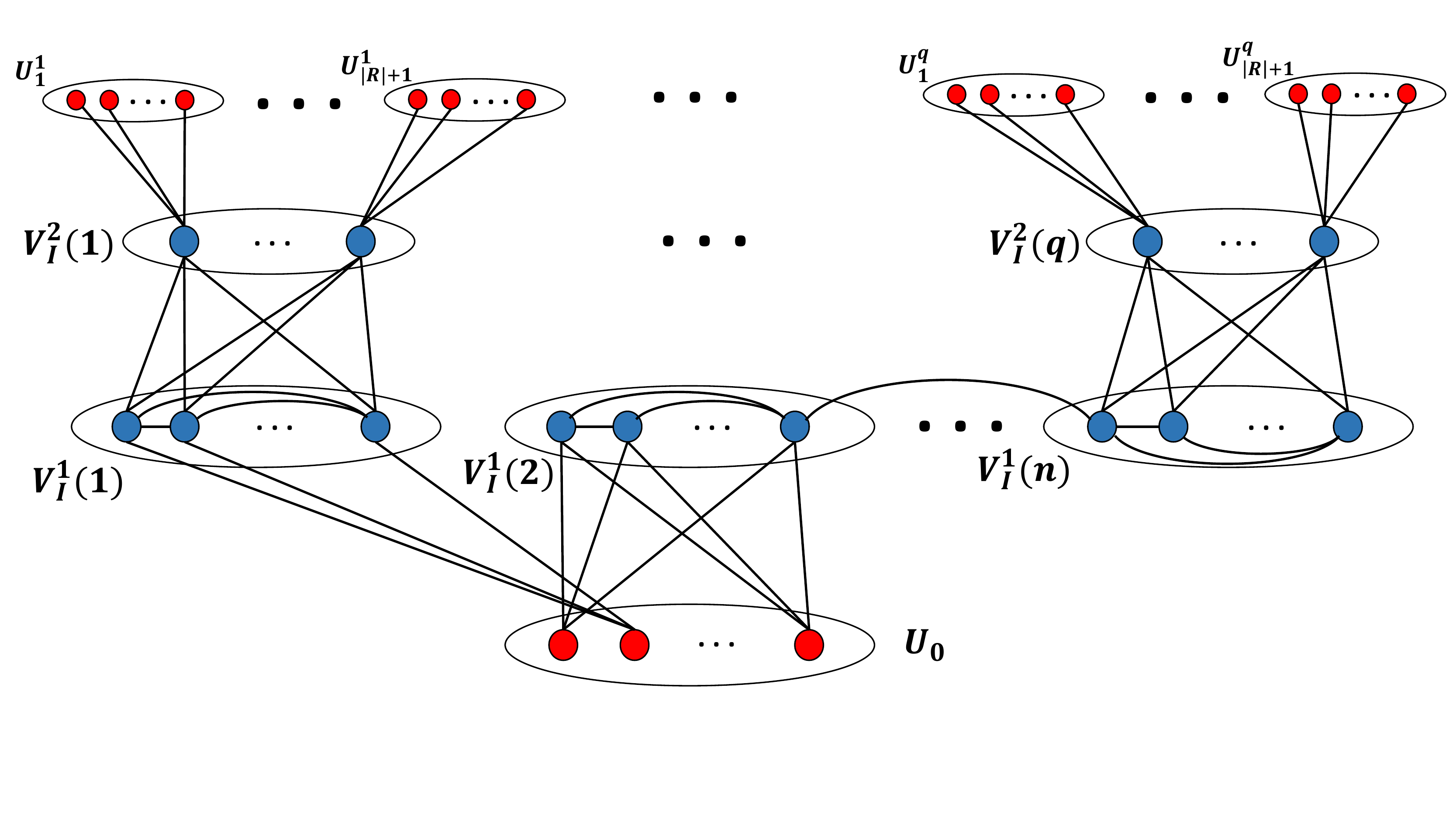}
	\vspace{-0.45in}
	\caption{Reduction from Red-Blue Set Cover to \ISI{} in which infected nodes are blue-colored and uninfected nodes are in red.}
	\label{fig:reduce}
\end{figure}

Given an instance of Red-Blue Set Cover with two sets $R, B$ and a family $T$, we suppose all the subsets in $T$ contains at least a blue element, otherwise we can trivially discard all those subsets since we never select that type of subsets. In the reverse way, we also suppose every pair of subset in $T$ has at least one red element different from each other. Otherwise we always select/reject both at the same time without changing the cost, in other words, we can group together to create one subset. We construct a corresponding \ISI{} instance consisting of the node set $V$, the infected set $V_I \subseteq V$ and the set of edges $E$ as follows (depicted in Fig.~\ref{fig:reduce}):
\begin{itemize}
	\item Set of infected nodes $V_I$: For each subset $T_i \in T$, there is a set $V_I^1(i)$ of infected nodes whose number is the number of blues in $T_i$. For each blue node $B_j$ in $B$, we form a set $V_I^2(j)$ of $|R|+1$ infected nodes.
	\item Set of uninfected nodes $V \backslash V_I$: For each infected node $l$ in $V_I^2(j)$, a set $U_l^j$ of $|R|+1$ uninfected nodes is constructed. We also have a set $U_0$ of $p$ uninfected nodes corresponding to the red set $R$ in Red-Blue Set Cover instance.
	\item Set of edges $E$: For any pair $(u,v) \in V_I^1(i)$, we connect them by an edge, so that the subgraph of nodes in $V_I^1(i)$ is a clique. For each $u \in T_i \cap T_j$, we connect the two corresponding nodes in $V_I^1(i)$ and $V_I^1(j)$ by an edge. For each $u \in V_I^1(i)$, we connect $u$ to all $|R|+1$ nodes in $V_I^2(u)$ and, subsequently, each node $l$ in $V_I^2(u)$ is connected to all $|R|+1$ nodes in $U_l^u$. For any pair $u,v \in V_I^1(j)$ for each $j \in \{1,...,n\}$, we connect $u$ with all the nodes in $V_I^2(v)$ and $v$ with all the nodes in $V_I^2(u)$. If the subset $T_i$ contains red element $R_j$, then for each $u \in V_I^1(i)$, there is an edge connecting $u$ to the corresponding node of $R_j$ in $U_0$.
\end{itemize}

Now, we will prove the two conditions consecutively. Our proof relies on two observations: the first one is that if the feasible solution for \ISI{} contains at least an infected node from $V_I^2(j)$ for some $j \in \{1,...,q\}$, then the number of uninfected nodes covered is at least $|R|+1$ which causes the cost to be at least $|R|+1$. The same phenomenon happens if an infected node $v$ in $V_I^1(j)$ for some $j \in \{1,...,n\}$ is not covered since there would be $|R|+1$ infected nodes in $V_I^2(v)$ not covered. On the other hand, if all the infected nodes in $V_I^1(j)$ for all $j \in \{1,...,n\}$ are covered, then all infected nodes in the whole network are indeed covered and at most $|R|$ uninfected nodes (in $U_0$) are also covered. The second observation with the previous case is that in the original Red-Blue Set Cover instance, we select those subset $T_i$ such that the corresponding $V_I^1(i)$ contains a infection source chosen in \ISI{}, then the cost in the two problem are equal (cover the same number of red elements/uninfected nodes).

\textit{Prove condition \textbf{(1)}}. Based on our observation, the optimal solution $S^*$ of the \ISI{} instance has to cover all the nodes in $V_I^1(j)$ for all $j$ and has the least number of uninfected nodes covered. From this solution, we construct the solution for the original Red-Blue Set Cover instance by selecting the subfamily $C^*$ of subsets $T_i$ such that the corresponding $V_I^1(i)$ contains a infection source in the optimal solution of \ISI{}. First, this subfamily covers all the blue elements since each blue element corresponds to some infected nodes in $V_I^1(j)$ for some $j$. Secondly, if this subfamily has the lowest cost (covers the least number of red elements). Otherwise, suppose that a different subfamily $\hat C$ has lower cost, then we can equivalently find another solution for the reduced \ISI{} instance and obtain the same cost (lower than that of $S^*$). That contradicts with the optimality of $S^*$.

\textit{Prove condition \textbf{(2)}}. Based on condition \textbf{(1)} that the optimal solution of \ISI{} instance infers the optimal solution of Ref-Blue Set Cover with the same cost. Suppose we have an $O(2^{\log^{1-\epsilon}n})$-approximate solution $\hat S$ for Red-Blue Set Cover instance, there are two possible cases:
\begin{itemize}
	\item If $\hat S$ contains a node in $V_I^2(j)$ for some $j$ or $\hat S$ does not cover a node in $V_I^1(j)$ for some $j$, then based on the first observation, the cost of $\hat S$ has to be at least $|R| + 1$. Because this is an $O(2^{\log^{1-\epsilon}n})$-approximate solution, we just select the whole family $T$ in Red-Blue Set Cover instance which has cost of only $|R|$ and obtain an $O(2^{\log^{1-\epsilon}n})$-approximate solution.
	\item Otherwise, based on the second observation, we can easily construct a solution for Red-Blue Set Cover with equal cost and thus obtain an $O(2^{\log^{1-\epsilon}n})$-approximate solution.
\end{itemize}

Lastly, note that the number of blue and red elements must be at least $|T|$, otherwise we can drop or merge some sets together without effecting any solution. Thus, by following our construction of the \ISI{} instance, we determine the number of uninfected nodes,
\begin{align}
	|V| &= |U_0| + \sum_{i=1}^{n}|V_I^1(i)| + \sum_{i=1}^{q} |V_I^2(i)|\cdot\sum_{j = 1}^{|R+1|} |U^i_j| \nonumber \\
	& \leq |T| + |T| + |T|^2(|T|+1) \leq |T|^4 \text{ \ \ \ }(\text|T| \geq 4)
\end{align}

Since $|T| = n$ and the Red-Blue Set Cover cannot be approximated within a factor of $O(2^{\log^{1-\epsilon}N})$ where $N = n^4$ for any $\epsilon > 0$ unless NP $\subseteq$ DTIME($N^{polylog(N)}$), we obtain our results in Theo.~\ref{theo:inapprox}.
\end{proof}

\vspace{-0.05in}
\section{Sampling-based SISI algorithm}
\label{sec:submodular}
In this section, we present \SISI{}, our sampling-based method with guarantee on achieving $\frac{2}{(1-\epsilon)^2}\Delta$-approximation factor for arbitrary small $\epsilon>0$. Here $\Delta$ equals the maximum nodes in $V_I$ that can infect a single node in the graph and is the same with the maximum sample size in Subsec.~\ref{subsec:ris}.

\textbf{Outline}. \SISI{} contains two key components: 1) an efficient Truncated Reverse Infection Sampling (\TRIS{}) to compute the objective with high accuracy and confidentiality (presented in Subsection~\ref{subsec:ris}) and 2) an innovative transformation of the studied problem into a submodular-cost covering problem to provide high quality solutions with performance guarantees (presented in Subsection~\ref{subsec:sub}). We show the combination of the two components to obtains the \SISI{} algorithm in Subsection~\ref{subsec:main_alg}.


\vspace{-0.05in}
\subsection{Truncated Reverse Infection Sampling}
\label{subsec:ris}
We propose the \emph{Truncated Reverse Infection Sampling} (\TRIS{}) strategy to generate random Reverse Reachable  (\RR{}) sets, following the \emph{reverse influence sampling} method (\RIS{}) pioneered in \cite{Borgs14}.
A \RR{} set, $R_{j}$,  is generated as follows.
\vspace{-0.03in}
\begin{Definition}[Reverse Reachable set (\RR{} set)]
	\label{def:ris}
	Given $G=(V,E)$, probability $\beta$ and propagation time $\tau$, a \RR{} set is generated from $G$ by 1) selecting a (uniformly) random source node $v \in V$, 2) generating a reverse random cascade from $v$ in $G$ within $\tau$ steps and 3) returning $R_{j}$ as the set of nodes in the cascade.
\end{Definition}
\vspace{-0.03in}

The main intuition is that each \RR{} set $R_j$ contains the nodes that can infect its source $v=\src(R_j)$ within a given time $\tau$. Thus \RR{} sets were used in previous works \cite{Borgs14,Tang15,Nguyen16} (without the step/time limit $t$)  to estimate influence of nodes. We shall show later in next subsection that \RR{} sets can also be fine-tuned to estimate the chance of being infection sources.

Note that the above description of generating \RR{} sets is model-independent, i.e., you can use it with many different cascade models for reverse cascade simulation in the step 2. For example, the reverse simulation for  \IC{} and \LT{}, the two most  popular cascades models, are presented in \cite{Borgs14} and \cite{Nguyen162}, respectively. Here we focus on the reverse sampling for \SI{} model and highlight the necessary changes to make the method work for our problem.

\begin{algorithm}
	\caption{\textsf{Fast-TRIS}}
	\label{alg:ris_fast}
	\KwIn{Graph $G$, probability $\beta$, max time $\tau$ and $V_I$}
	\KwOut{A random \TRR{} set $R_{j}$}
	Pick a random node $u \in V$\\
	\TRR{} set $R_j = \{u\}$\\
	Infection time $T\{v\} = \infty $, $\forall v \in V\backslash\{u\}$, $T\{u\} = 0$\\
	Min priority queue $PQ = \{u\}$\\
	\While{$PQ$ not empty}{
		$u = PQ.pop()$\\
		\ForEach{ $v \in (\text{in-neighbors}(u) \backslash R_{j}) \cup PQ$}{
			$r \leftarrow$ a random number in [0,1]\\
			$t \leftarrow \lceil \log_{1-\beta}(1-r)\rceil$ \text{  } \{Assume $0 < \beta < 1$\}\\
			$T(v) = \min\{T(v), T(u) + t\}$\\
			\If{$T(v) < \tau$}{
				\eIf{$v \notin R_j$}{
					\If{$v \in V_I$}{
						$R_j = R_j \cup \{v\}$\\
					}
					$PQ.push(v)$\\
				}{
				$PQ.update(v)$\\
			}
		}
	}
}
Return $R_{j}$\\
\end{algorithm}

\subsubsection{Generating \RR{} Sets under \SI{} model.} The main difference between \SI{} model vs. \LT{} and IC{} models are \SI{} model allows multiple attempts for an infected node to its neighbors in contrast to a single attempt in \IC{} and \LT{}. Given a network $G=(V, E)$ and infection  probability $0< \beta \leq 1$, \RR{} sets in the \SI{} model are generated as follows.
\begin{itemize}
	\vspace{-0.05in}
	\item[1)] Select a random node $u$. Only $u$ is infected at time $0$ and all other nodes are not infected.
	\vspace{-0.05in}
	\item[2)] For each time step $i \in [1,\tau]$, consider all edges $(u, v) \in E$ in which $v$ is infected and $u$ is not infected (note the direction). Toss a $\beta$-head biased coin to determine whether $u$ succeeds in infecting $v$. If the coin gives head (with a probability $\beta$), we mark $u$ as infected.
	\vspace{-0.05in}
	\item[3)] After $\tau$ steps, return $R_j$ as the set of infected nodes, \emph{removing all nodes that are not in $V_I$}.
	\vspace{-0.05in}
\end{itemize}
Note the last step, the nodes that are not in $V_I$ will be removed from the \RR{} set (hence the name truncated). This truncation is due to the observation that the suspected nodes must be among the infected nodes in $V_I$. Our \RR{} sets are in general smaller than the \RR{} sets in \cite{Borgs14} and might be empty. This saves us a considerable amount of memory in storing the \RR{} sets.

A naive implementation of the above reserve sampling has a high complexity and does not scale when $\tau$ grows, thus we present a fast implementation using geometric distribution in Algorithm \ref{alg:ris_fast}. 

The complete pseudocode for the fast \TRIS{} algorithm is described in Alg.~\ref{alg:ris_fast}. The key observation to speed up the \TRIS{} procedure is that each trial in the sequence of infection attempts is a Bernoulli experiment with success probability of $\beta$. Thus this sequence of attempts until successful actually follows a geometric distribution. Instead of tossing the Bernoulli coin many times until getting a head, we can toss once and use the geometric distribution to determine the number of Bernoulli trials until successful (Lines~8,9). 

Another issue is the order of attempts since a node can be infected from any of her in-neighbors but only the earliest one counts. Therefore, we will keep the list of all newly infected nodes in a min priority queue (PQ) w.r.t infection time. In each iteration, the top node  is considered (Lines~6). The algorithm behaves mostly like the legacy Dijkstra's algorithm \cite{Goldberg96} except we have time for a node $w$ to infect a node $v$ on each edge $(w, v)$ instead of the length. Also, the algorithm is constrained within the region consisting of nodes at most `distance' $\tau$ from the selected $u$.


The time complexities of the naive and fast implementation of \TRIS{} are stated in the following lemma.
\begin{Lemma}
	\label{lem:ris_com}
	Expected time complexity of the naive \TRIS{} is,
	\begin{align}
		\label{eq:ris_com1}
		C(R_j) = \frac{\Delta m \tau}{n}
	\end{align}
	and that of the fast implementation is,
	\begin{align}
		\label{eq:ris_fast}
		C'(R_j) = \frac{\Delta m}{n} + \Delta\log (\Delta)\log (1 + \frac{\Delta m}{n^2})
	\end{align}
\end{Lemma}
\begin{proof}
	Similar to the analysis of Expected Performance of Dijkstra's Shortest Path Algorithm in \cite{Goldberg96} and denote the expected complexity of the fast algorithm by $C'(R_j)$, we have,
	\begin{align}
		\label{eq:total_com}
		C'(R_j) = C(edges) + \Delta\log (\Delta)\log (1 + C(edges)/n)
	\end{align}
	where $C(edges)$ is the expected number of edges examined. Note that this is different from $C(R_j)$ since in this case, each edge can be checked once while, for the latter, it is multiple until successful. $\Delta$ is defined previously as the maximum size of a RR set. We also have,
	\begin{align}
		C(edges) \leq \frac{1}{n}\sum_{u \in V} \sum_{v \in V} \Pr[u,v] d^{in}(v) \leq \frac{\Delta m}{n}
	\end{align}
	in which the details are similar to that of Eq.~\ref{eq:ris_com1}. Thus, combining with Eq.~\ref{eq:total_com}, we obtain,
	\begin{align}
		C'(R_j) = \frac{\Delta m}{n} + \Delta\log (\Delta)\log (1 + \frac{\Delta m}{n^2})
	\end{align}
	
	In Eq.~\ref{eq:ris_fast}, the first term is usually the leading factor and, then, the complexity depends mostly on $\frac{\Delta m}{n}$. We now analyze the expected time complexity $C(R_j)$ of generating $R_j$ by the naive way.
	\begin{align}
		C(R_j) \leq \frac{\tau}{n}\sum_{u \in V} \sum_{v \in V} \Pr[u,v] d^{in}(v) = \frac{\tau}{n} \sum_{v\in V} d^{in}(v) \sum_{u \in V} \Pr[u,v] \nonumber
	\end{align}
	where $\Pr[u,v]$ is the probability of $v$ infected by $u$ within $\tau$ steps, $d^{in}(v)$ is the in-degree of $v$. Here we take the average over all possible sources $u$ of $R_j$ (each has probability $1/n$) and the maximum number of edge checks for node $v$ is $\tau d^{in}(v)$. Let denote the maximum size of a random \RR{} set as $\Delta$, we get $\sum_{u \in V} \Pr[u,v] \leq \Delta$ and thus,
	\begin{align}
		C(R_j) \leq \frac{\tau}{n} \sum_{v\in V} d^{in}(v) \Delta = \frac{\Delta\tau}{n} \sum_{v\in V} d^{in} (v) = \frac{\Delta m \tau}{n}
	\end{align}

	From Eq.~\ref{eq:ris_com1}, the complexity depends linearly on $\tau$ and is very high with large values of $\tau$.
\end{proof}

Thus, the running time $C'(R_j)$ of our fast implementation  is roughly \emph{$\tau$ times smaller than that $C(R_j)$} of the naive implementation, especially, for large values of $\tau$.

\subsubsection{Chance of Being Infection Sources}
We show how to utilize the generated \RR{} sets to estimate the chance that nodes being infection sources. First, we classify each generated \RR{} $R_j$ into one of the two groups, based on the source of $R_j$, denoted by $\src(R_j)$.
\begin{itemize}
  \item $\R_{Blue} = \{R_j | \src(R_j) \in V_I\}$: The set of blue \RR{} sets that sources are.
  \item $\R_{Red} =  \{R_j | \src(R_j) \notin V_I\}$: The set of red \RR{} sets that sources are  \emph{not} in $V_I$.
\end{itemize}

Since the infection sources infect the nodes in $V_I$ but not the nodes outside of $V_I$ (within a time $\tau$), thus, the infection sources should appear \emph{frequently in blue} \RR{} sets (of which sources are in $V_I$) and appear \emph{infrequently in red} \RR{} sets (of which sources are not in $V_I$.) Thus, a node $v$ that appear in many blue \RR{} and few red \RR{} sets will be more likely to be among the infection sources.

The above observation can be generalized for a given a subset of nodes $S \subset V_I$, e.g., a subset of suspected nodes. A subset $S$ that \textit{covers} (i.e. to intersect with) \emph{many blue} \RR{} sets and \emph{few red} \RR{} sets  will be more likely to be the infection sources.

Define the following two subgroups of \RR{} sets,
\begin{align}
\R_{Blue}^{-}(S) &=\{R_j |  R_j  \in \R_{Blue} \text{ and } R_j \cap S \mathbb{=} \emptyset\}, \text{ and }\\
\R_{Red}^{+}(S) &=\{R_j | R_j \in \R_{Red} \text{ and } R_j \cap S \mathbb{\neq} \emptyset\}.
\end{align}
They are the blue \RR{} sets that a suspected subset $S$ ``fails'' to cover (i.e. to intersect with) and the red \RR{} sets that $S$ (``mistakenly'') covers. The less frequent a random \RR{} set $R_j$ falls into one of those two subgroups, the more likely $S$ will be the infection sources.


Formally, we can prove that the probability of a random \RR{} set falls into one of those two subgroups equals exactly our objective function, denoted by $\mathbb{E}[D(\hat S, \tau,V_I)]$. We state the result in the following lemma.

\newcommand{\twopartdef}[3]
{
	\left\{
	\begin{array}{ll}
		#1 & \mbox{if } #2 \\
		#3 & \mbox{otherwise}
	\end{array}
	\right.
}
\newcommand{\threepartdef}[5]
{
	\left\{
	\begin{array}{lll}
		#1 & \mbox{if } #2 \\
		#3 & \mbox{if } #4 \\
		#5 & \mbox{otherwise }
	\end{array}
	\right.
}
\begin{Lemma}
	\label{lem:ris}
	Given a fixed set $S \in V_I$, for a random \RR{} set $R_{j}$, denote $X_j$ a random variable such that,
	\begin{align}
		\label{eq:eq14}
		X_j = \twopartdef {1} {R_j \in  \R^-_{Blue}(S) \text{ or } R_j \in \R^+_{Red}(S)} {0}
	\end{align}
	then,
	\begin{align}
		\label{eq:lem_ris}
		\mathbb{E}[X_j] = \frac{\mathbb{E}[D(S,\tau,V_I)]}{n}
	\end{align}
\end{Lemma}

\begin{proof}
	 Since for a random RR set $R_j$, $R_j \in \R^{-}_{Blue}(S)$ and $R_j \in \R^{+}_{Red}(S)$ are two mutually exclusive events,
	 \begin{align}
		\label{eq:ori}
	 	\mathbb{E}[X_j] = & \Pr_{R_j}[R_j \in \R^{-}_{Blue}(S)] + \Pr_{R_j}[R_j \in \R^{+}_{Red}(S)]
	 \end{align}
	 We will prove an equivalent formula of Eq.~\ref{eq:lem_ris} that,
	 \begin{align}
	 	\E[D(S,\tau)] = n(\Pr_{R_j}[R_j \in \R^{-}_{Blue}(S)] + \Pr_{R_j}[R_j \in \R^{+}_{Red}(S)]) \nonumber
	 \end{align}

	 Let define $\mathcal{G}$ as a realization of the graph $G$, $\mathcal{G} \sim G$, where each edge $(u,v)$ is assigned a length value indicating the number of trials $u$ has to make until $v$ gets infected from $u$. In one realization $\mathcal{G}$, the cascade from $S$ at time $\tau$, $V(S,\tau)$, is uniquely defined (the reachable nodes from $S$ within $\tau$-length path) and so as $D(S,\tau)$. According to the definition of $\E[D(S,\tau)]$ in Eq.~\ref{eq:diff_2}, we have,
	 \begin{align}
		 \E[D(S,\tau)] = \sum_{u \in V_{I}} \Pr_{\mathcal{G}\sim G}[u \text{ not infected}] + \sum_{v \notin V_{I}}\Pr_{\mathcal{G}\sim G}[v \text{ infected}] \nonumber
	 \end{align}
	 
	 Let denote $R_{j}(u)$ be a random \RR{} set rooted at $u$, the first term in the right-hand side is equivalent to,
	 \begin{align}
	 	\sum_{u \in V_{I}} \Pr_{\mathcal{G}\sim G}[u \text{ not infected}] = \sum_{u \in V_{I}} \Pr_{R_{j}(u) \vdash \mathcal{G}}[S \cap R_{j}(u) = \emptyset] \nonumber
	 \end{align}
	 \noindent where $R_{j}(u) \vdash \mathcal{G}$ denotes the consistency of $R_{j}(u)$ to $\mathcal{G}$ since $\mathcal{G}$ is a realization of $G$ and thus $R_{j}(u)$ is  well-defined. Since,
	 \begin{align}
	 	\Pr_{R_{j}(u) \vdash \mathcal{G}}[S \cap R_{j}(u) = \emptyset] = \Pr_{R_{j}}[S \cap R_{j} = \emptyset \text{ }|\text{ } \src(R_{j}) = u] \nonumber
	 \end{align}
	 we obtain,
	 \begin{align}
	 	\sum_{u \in V_{I}}& \Pr_{\mathcal{G}\sim G}[u \text{ not infected}] = \sum_{u \in V_{I}} \Pr_{R_{j}}[S \cap R_{j} = \emptyset \text{ }|\text{ } \src(R_{j}) = u] \nonumber \\
	 	& = \sum_{u \in V_{I}} \frac{\Pr_{R_{j}}[S \cap R_{j} = \emptyset \text{ }\&\text{ } \src(R_{j}) = u]}{\Pr_{R_{j}}[\src(R_{j}) = u]} \nonumber \\
	 	& = \sum_{u \in V_{I}} \Pr_{R_{j}}[S \cap R_{j} = \emptyset \text{ }\&\text{ } \src(R_{j}) = u]\cdot n \nonumber \\
	 	& \text{ (since the source of each \RR{} set is randomly chosen)} \nonumber \\
	 	& = n\sum_{u \in V_{I}} \Pr_{R_{j}}[S \cap R_{j} = \emptyset \text{ }\&\text{ } \src(R_{j}) = u] \nonumber \\
	 	\label{eq:eq30}
	 	& = n\Pr_{R_{j}}[S \cap R_{j} = \emptyset \text{ }\&\text{ } \src(R_{j}) \in V_{I}] \\
	 	& = n \Pr_{R_{j}}[R_j \in \R^{-}_{Blue}(S)]
	 \end{align}
	 \noindent The Eq.~\ref{eq:eq30} follows from the fact that, for all $u \in V_{I}$, $(S \cap R_{j} = \emptyset \text{ }\&\text{ } \src(R_{j}) = u)$ are mutually exclusive. Thus,
	 \begin{align}
		 \label{eq:eq1}
		 \sum_{u \in V_{I}} \Pr_{\mathcal{G}\sim G}[u \text{ not infected}] = n \Pr_{R_{j}}[R_j \in \R^{-}_{Blue}(S)]
	 \end{align}
	 Similarly, we can also achieve,
	 \begin{align}
		 \label{eq:eq2}
		 \sum_{v \in \bar{V}_{I}} \Pr_{\mathcal{G}\sim G}[v \text{ infected}] = n\Pr_{R_j}[R_j \in \R^{+}_{Red}(S)]
	 \end{align}
	 From Eq.~\ref{eq:eq1}, Eq.~\ref{eq:eq2} and Eq.~\ref{eq:ori}, we obtain
	 \begin{align}
		 \E[D(S,\tau)] = n(\Pr_{R_j}[R_j \in \R^{-}_{Blue}(S)] + \Pr_{R_j}[R_j \in \R^{+}_{Red}(S)]) \nonumber
	 \end{align}
	 which completes the proof of Lem.~\ref{lem:ris}.
\end{proof}

Lem.~\ref{lem:ris} suggests a two-stages approach to identify the infection sources: 1) generating many \RR{} sets and 2) look for a subset $S\subset V_I$ that minimize the size of $| \R^-_{Blue}(S) \cup R_j \in \R^+_{Red}(S)|$. In next two subsections, we address two key issues of this approach  1) \emph{Optimization method} to identify $S$ with guarantees and 2) \emph{Sample complexity}, i.e., how many \RR{} sets is sufficient to generate a good solution. Too few \RR{} sets  lead to biased and poor solutions, while too many \RR{} set lead to high running time.

%
\vspace{-0.03in}
\subsection{Submodular-cost Covering}
\label{subsec:sub}
We will transform the \ISI{} problem to a submodular-cost covering problem over the generated \RR{} sets. This allows us to apply the $\Delta$-approximation algorithm in \cite{Koufogiannakis13}, where $\Delta$ is the maximum size of any \TRR{} set.

 By Lemma \ref{lem:ris}, the problem of minimizing $\E[D(S,\mathcal{M},V_I)]$ can be cast as a minimization problem of $\Pr[R_j \in \R^{-}_{Blue}(S) \cup \R^{+}_{Red}(S)]$. This, in turn, can be approximated with the following problem over the generated \RR{} sets.
\vspace{-0.05in}
\begin{align}
	\min_{S \subseteq V_I} |\R^-_{Blue}(S) \cup \R^{+}_{Red}(S)|,
\end{align}
\vspace{-0.12in}

\noindent and, since $R_j \in \R^-_{Blue}(S)$ and $R_j \in \R^{+}_{Red}(S)$ are disjoint, the above minimization problem is equivalent to,
\vspace{-0.18in}

\begin{align}
	\label{eq:eq4}
	\min_{S\subseteq V_I} |\R^-_{Blue}(S)| + |\R^{+}_{Red}(S)|
\end{align}
\vspace{-0.12in}

We shall convert the above problem to the \emph{submodular-cost covering} in \cite{Koufogiannakis13}, stated as follows.
\begin{Definition}[Submodular-cost covering]\cite{Koufogiannakis13}
An instance is a triple $(c, \mathcal C, U)$ where
\begin{itemize}
\item The cost function $c(x): \mathbf{R}^n_{\geq 0} \rightarrow \mathbf{R}_{\geq 0}$ is submodular, continuous, and non-decreasing.
\item The constraint set $\mathcal C\subseteq 2^{\mathbf{R}_{\geq 0}}$ is a collection of covering constraints, where each constraint $S \in \mathcal C$ is a subset of $\mathbf{R}^n_{\geq 0}$.
\item For each $j \in [n]$, the domain $U_j$ for variable $x_j$ is any subset of $\mathbf{R}_{\geq 0}$.
\end{itemize}
The problem is to find $x \in \mathbf{R}^n_{\geq 0}$, minimizing $c(x)$ subject to $x_j \in U_j, \forall j \in [n]$  and $x \in S, \forall S \in \mathcal C$.
\end{Definition}

\captionsetup{labelfont=bf}
\begin{figure}[!ht] \centering
	\includegraphics[width=0.7\linewidth]{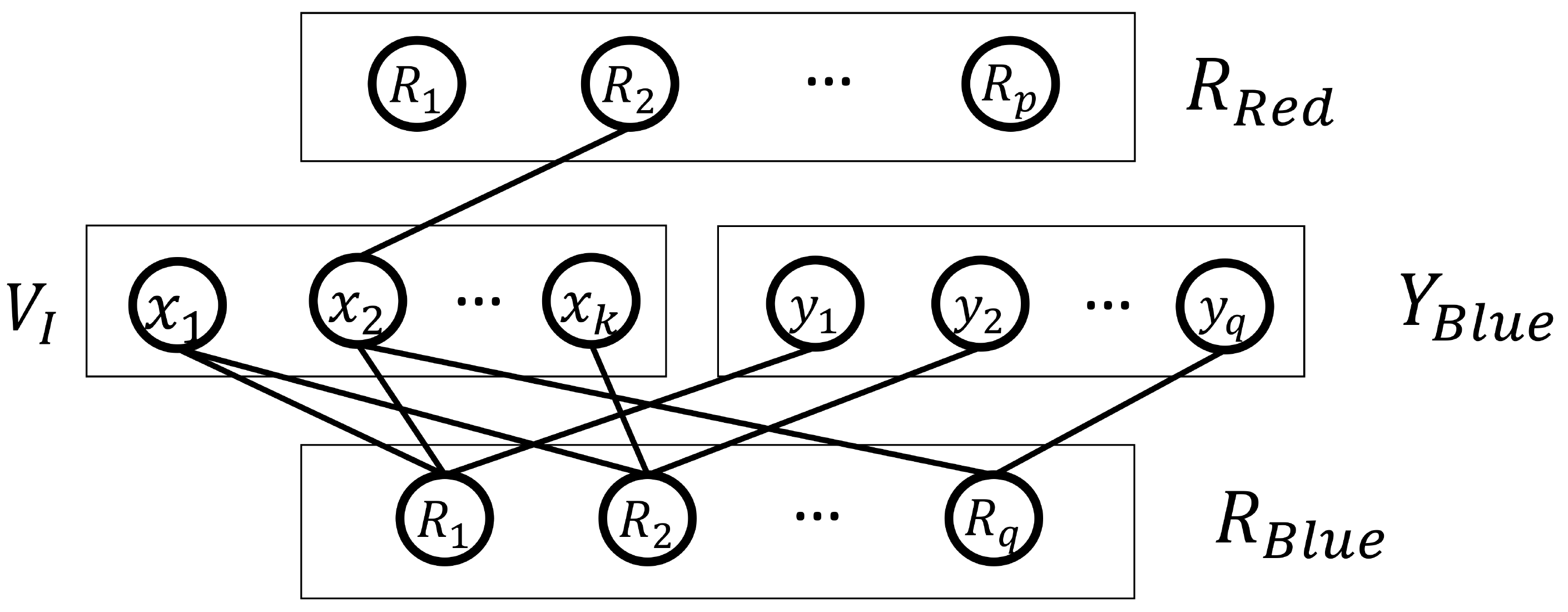}
	\caption{Conversion to Submodular-cost covering.}
	\label{fig:convert}
\end{figure}

\textbf{Conversion to submodular-cost covering problem.}
We convert the form in Eq.~\ref{eq:eq4} into a submodular-cost covering problem as demonstrated in Fig.~\ref{fig:convert}. Let $q = |\R_{Blue}|$ and $p = |\R_{Red}|$. 
We associate a variable $x_u \in [0, 1]$ for each $u \in V_{I}$ to indicate whether the corresponding node is selected as an  infected source. We also assign a variable $y_j$ to each \TRR{} set $R_j \in \R_{Blue}$. We require all blue \RR{} $R_j$ sets to be covered through the constraint $\max\{\max_{u \in R_j}x_u, y_j\} \ge 1$. Thus for each blue $R_j$ either $x_v=1$ for some $v \in R_j$  or the corresponding $y_j=1$. 

The objective is to minimize the cost function $\min_{x,y} c(x,y) = \sum_{R_j \in \R_{Red}} \max_{u\in R_j}(x_u) + \sum_{j=1}^{q} y_j$. The first part of the cost function $\max_{u\in R_j}(x_u)$
is a submodular function since the $\max$ function is submodular (see footnote 1, page 2 in \cite{Koufogiannakis13}). The second part $ \sum_{j=1}^{q} y_j$ is a linear function, and thus is also a submodular function. Therefore, \emph{the objective is a submodular function}.

Thus, the problem in Eq.~\ref{eq:eq4} can be converted to the following submodular-cost covering problem,
\begin{align}
	\label{eq:eq5}
	&\quad \quad \min_{x,y} c(x,y) = \sum_{R_j \in \R_{Red}} \max_{u\in R_j}(x_u) + \sum_{j=1}^{q} y_j\\
	&\text{subject to (for each } R_j \in \R_{Blue} \text{) } \max\{\max_{u \in R_j}x_u, y_j\} \ge 1 \nonumber
\end{align}

\noindent Since for any assignment of variable set $x$, we have a corresponding source selection: node $u$ is selected as infection source if $x_u = 1$. The first term $\sum_{R_j \in \R_{Red}} \max_{u\in R_j}(x_u)$ in Eq.~\ref{eq:eq5} is equivalent to $| \R^{+}_{Red}(S) |$ in Eq.~\ref{eq:eq4} and similarly $\sum_{j=1}^{q} y_j$ together with the constraints is equivalent to $|\R^{-}_{Blue}(S)|$. In Eq.~\ref{eq:eq5}, each covering constraint is associated with a blue \TRR{} set $R_j$ and says that if $R_j$ is not covered by any variable $x_u$ ($x_u = 1$), then $y_j = 1$ which will increase the cost function by 1. Thus, Eq.~\ref{eq:eq5} minimizes the number of red \TRR{} sets covered and blue \TRR{} sets uncovered.
\setlength{\textfloatsep}{3pt}
\begin{algorithm}
	\caption{\textsf{Submodular-cost-Covering}}
	\label{alg:sub}
	\KwIn{Infected set $V_{I}$, collection of \TRR{} sets $\R$}
	\KwOut{An $\Delta$-approximate set $\hat S$}
	Formulate the submodular cost covering version from $\R$\\
	$x_u = 0 , \forall u \in V_I$ and $y_j = 0, \forall j : R_j \in \R_{Blue} $\\
	\ForEach{$R_j \in \R_{Blue}$}{
		$\theta = \displaystyle\min_{u \in R_j}\sum_{R_{t}\in \R^+_{Red}(u)}(1 - \displaystyle\max_{v \in R_{t}}x_v)$\\
		$\theta = \min\{\theta, 1-y_j \}$\\
		\ForEach{$u \in R_j$}{
			\eIf{ $\R^+_{Red}(u) = \emptyset$} {
				$x_u =1$
			}{
			$x_u = \frac{1}{|\R^+_{Red}(u)|}\big(\theta + \displaystyle\sum_{R_{t}\in \R^+_{Red}(u)}\displaystyle\max_{v \in   R_{t} }x_v \big)$
		}
	}
	$y_j = y_j + \theta$\\
}
Add $u$ into $\hat S$ if $x_u = 1$\\
Return $\hat S$\\
\end{algorithm}
\textbf{$\Delta$-Approximation Algorithm.}
Our reformulation of \ISI{} to submodular-cost covering problem is similar to that of the facility location problem in Section 7 of \cite{Koufogiannakis13}. According to Lemma 5 in \cite{Koufogiannakis13}, the following greedy algorithm (Alg.~\ref{alg:sub}) runs in linear time with respect to the total size of all the \TRR{} sets and returns an $\Delta$-approximate solution.
\begin{theorem}
	\label{theo:sub}
	Alg.~\ref{alg:sub} returns an $\Delta$-approximate solution for the submodular-cost covering formulation of the \ISI{} problem, where $\Delta$ is the maximum size of an \TRR{} set (thus, $\Delta \leq V_I$), and runs in linear time.
\end{theorem}

The Alg.~\ref{alg:sub} starts with formulating the submodular-cost covering problem from $V_I$ and $\R$ by creating the necessary variables, cost function and constraints as specified previously. A variable $x_u$ is initialized to 0 and gets updated in the iterations that node $u$ is in the \TRR{} set considering in those iterations. The algorithm passes through all the \TRR{} sets $R_j \in \R_{Blue}$ and makes each of them satisfied in a single iteration in which it calculates the minimum increase $\theta$ of the cost function (Line 4-5) that satisfies the constraint. This minimum increase is computed by sequentially trying to raise each variable $x_u : u \in R_j$ or $y_j$ to 1 (covering) and calculating the corresponding cost. Afterwards, it updates each variable of $R_j$ by an amount that makes the cost function increased by $\theta$ (Line 6-11). At the end, it selects the nodes in $V_{I}$ that have value 1 in their variables (Line 12).

\begin{algorithm}
	\caption{\SISI{} Algorithm}
	\label{alg:sisi}
	\KwIn{Graph $G=(V, E)$, infection probability $\beta$, a set of infected nodes $V_I$, an infection model $\mathcal{M}$ and $\epsilon, \delta \in (0, 1)$.}
	\KwOut{Initial infected set $\hat S$.}
	$\Lambda = (1 + \epsilon)2c \Big [ \ln \frac{2}{\delta} + k\ln 2  + 1\Big ] \frac{1}{\epsilon^2}$\\
	$T = \Lambda, \R \leftarrow \emptyset$\\
	\Repeat{$|R_{Blue}^{-}(\hat S)| + |R_{Red}^{+}(\hat S)| \ge \Lambda$}{
		Generate $T$ additional \TRR{} sets  by \textsf{Fast-TRIS} (or the reverse sampling in \cite{Borgs14,Nguyen16} for \IC{}, \LT{} models)\\
		$\hat S = $ \textsf{Submodular-cost-Covering}($V_I, \R$)\\
		$T = |\R|$\\
		$\Delta = \max_{R_j} |R_j|$\\
		\If{$\epsilon > 1/(1+\Delta)$}{
			$\epsilon = 1/(1+\Delta)$\\
			$\Lambda = (1 + \epsilon)2c \Big [ \ln \frac{2}{\delta} + k\ln 2  + 1\Big ] \frac{1}{\epsilon^2}$\\
		}
	}
	Post-optimization($\hat S$)\\
	Return $\hat S$\\
\end{algorithm}

\vspace{-0.05in}
\subsection{SISI Approximation Algorithm}
\label{subsec:main_alg}

\label{subsec:alg}
We will describe the approximation algorithm, named \SISI{}, which combines the three key advanced components: \TRIS{} sampling (Subsec.~\ref{subsec:ris}), the $\Delta$-approximate submodular-cost covering algorithm (Subsec.~\ref{subsec:sub}) and a stopping condition in \cite{Nguyen16}, to solve the \ISI{} problem and returns an $\Delta\frac{2}{(1-\epsilon)^2}$-approximate solution with at least $(1-\delta)$-probability (proved in Sec.~\ref{sec:approx}). The description of \SISI{} is given in Alg.~\ref{alg:sisi}.

\SISI{} begins with initializing $\Lambda$ which will decide the stopping condition (Line~11). The whole algorithm iterates through multiple steps: in the first step, it generates $\Lambda$ \TRR{} sets and add them to $\R$ since, to satisfy the stopping condition (Line~11), we need at least $\Lambda$ \TRR{} sets; in subsequent iterations, the algorithm doubles the number of \TRR{} sets in $\R$ by generating $|\R|$ more. In each iteration, it utilizes the submodular-cost covering algorithm to find the candidate set $\hat S$ (Line~5) and check whether we have sufficient statistical evidence to achieve a good solution by checking the \textit{stopping condition} (Line~11). The stopping condition plays a decisive roles in both theoretical solution quality and the complexity of the algorithm. The condition in \SISI{} is derived from the results of optimal sampling for Monte-Carlo estimation studied in \cite{Dagum00}. In the next section, we will prove that with this stopping condition, \SISI{} returns an $\Delta\frac{2}{(1 - \epsilon)^2}$-approximate solution with probability of at least $(1-\delta)$, where $\epsilon, \delta$ are given as inputs. The check in Lines~8-10 is to guarantee $\epsilon$ small enough and described in Sec.~\ref{sec:approx}. At the end of the algorithm, \SISI{} performs a post-optimization of $\hat S$ which incrementally removes nodes in $\hat S$ if that improves the objective function. 

\section{Algorithm Analysis}
\label{sec:approx}
We will analyze the approximation guarantee and time complexity of \SISI{} algorithm. In short, we prove that \SISI{} returns an $\Delta\frac{2}{(1-\epsilon)^2}$-approximate solution. In the sequel, we will present the time complexity of \SISI{}.
\subsection{Approximation Guarantee}
\label{subsec:approx}
To prove the approximation guarantee of \SISI{}, we show two intermediate results: 1) with $\frac{n\Lambda}{\E[D(\hat S)]}$ \TRR{} sets where $\hat S$ is the solution returned by \SISI{}, $\E[D(\hat S)] = \E[D(\hat S, \mathcal{M},V_I)]$ for short since the $\mathcal{M}, V_I$ are fixed, all the sets $S \subset V_{I}$ are well approximated from $\R$ with high probability (Lem.~\ref{lem:app}) and 2) the actual number of \TRR{} sets generated in \SISI{} is greater than $\frac{n\Lambda}{\E[D(\hat S)]}$ with high probability (Lem.~\ref{lem:rr}). Then, combine these results and the property of submodular-cost covering, we obtain the approximation factor in Theo.~\ref{theo:approx}.
%
%

Denote $D_{\R}(S) = \frac{n}{|\R|}(|\R_{Blue}^{-}(S)| + |\R_{Red}^{+}(S)|)$, which is an approximation of $\E[D(S)]$, achieved from the collection of \TRR{} sets $\R$. The following lemma states the approximation quality of a set $S \subseteq V_{I}$. We assume that $\E[D(\hat S)] \neq 0, \forall \hat S \subset V_I$ since the case of equaling 0 only happens if $V_I$ is a disconnected clique with edge weights being all 1 and then, every set $S \in V_I$ are exactly identical. In that case,  the sources can be any set of nodes and are intractable to identify.
\begin{Lemma}
	\label{lem:app}
	If we have $T^*=\frac{n\Lambda}{\E[D(\hat S)]}$ \TRR{} sets where $\hat S$ is the solution returned by \SISI{}, then for a set $S \subseteq V_{I}$,
	\begin{align}
		\Pr[|D_{\R}(S) - \E[D(S)]| \geq \epsilon \sqrt{\E[D(S)] \cdot \E[D(\hat S)]}] \leq \frac{\delta}{M} \nonumber
	\end{align}
	where $M = 2^k + 1$ and $k = |V_I|$.
\end{Lemma}
\begin{proof}
	First, for a subset $S \subseteq V_I$ and a random \TRR{} set $R_j$, recall the binary random variable $X_j$ in Eq.~\ref{eq:eq14} that,
	\begin{align}
		X_j = \begin{cases} 1 &\mbox{if } R_j \in \R_{Blue}^{-}(S) \cup \R_{Red}^{+}(S) \\
		0 & \mbox{otherwise}. \end{cases}
	\end{align}
	Thus, the series of \TRR{} sets in $\R$ corresponds to a sequence of samples of $X_j$, denoted by $\{X^1_j,X^2_j,\dots\}$. Intuitively, since the \TRR{} are generated independently, the resulted sample sequence of $X_j$ should also be independent and identically distributed in $[0,1]$. However, similar to the Stopping Rule Algorithm in \cite{Dagum00} that \SISI{} creates a dependency on the samples by stopping the algorithm when some condition is satisfied. \SISI{} jumps to the next round when $|R_{Blue}^{-}(\hat S)| + |R_{Red}^{+}(\hat S)| \ge \Lambda$ or $\sum_{i = 1}^{|\R|} X_j^i \ge \Lambda$ is not met and hence, whether we generate more samples depending on the current set of \TRR{} sets. Interestingly, similar to the case of Stopping Rule Algorithm in \cite{Dagum00}, the sequence $\{X^1_j,X^2_j,\dots\}$ forms a \textit{martingle} and the following results follow from \cite{Dagum00}:
	
	Let $X_j^1, X^2_j,...$ samples according to $X_j$ random variable in the interval $[0, 1]$ with mean $\mu_{X_j}$ and variance $\sigma_{X_j}^2$ form a martingale and $\hat \mu_{X_j} = \frac{1}{T}\sum_{i=1}^{T}X^i_j$ be an estimate of $\mu_{X_j}$, for any fixed $T > 0, 0 \geq \epsilon \geq 1$,
	\begin{align}
		\label{eq:chef_1}
		\Pr[\hat \mu_{X_j} \geq  (1+ \epsilon) \mu_{X_j}] \leq e^{\frac{-T\mu_{X_j}\epsilon^2}{2c}}
	\end{align}
	and,
	\begin{align}
		\label{eq:chef_2}
		\Pr[\hat \mu_{X_j} \leq  (1- \epsilon) \mu_{X_j}] \leq e^{\frac{-T\mu_{X_j}\epsilon^2}{2c}}.
	\end{align}
	
	Recall that the value of $D_{\R}(S)$ is equivalent to,
	\begin{align}
		D_{\R}(S) = \frac{n}{|\R|}\sum_{i = 1}^{|\R|} X^i_j
	\end{align}

	Denote $\hat \mu_S = \frac{1}{|\R|}\sum_{i = 1}^{|\R|} X^i_j$ which is an estimate of $\mu_S = \frac{1}{n}\E[D(S)]$, then $T^* = \frac{\Lambda}{\mu_{\hat S}}$ and the inequality in Lem.~\ref{lem:app} can be rewritten,
	\begin{align}
		\label{eq:eq20}
		\Pr[|\hat \mu_S - \mu_S| \geq \epsilon \sqrt{\mu_{\hat S}\mu_S}] \leq \frac{\delta}{M}
	\end{align}
	
	Now, apply the inequality in Eq.~\ref{eq:chef_2} on the left side of the above Eq.~\ref{eq:eq20}, we have,
	\begin{align}
		\Pr[\hat \mu_S \leq  (1 - \epsilon\sqrt{\frac{\mu_{\hat S}}{\mu_S}}) \mu_S] \leq e^{\frac{-T^*\mu_S\epsilon^2\mu_{\hat S}}{2c\mu_S}} = e^{-(\ln(2/\delta) + k\ln 2 + 1)} \nonumber
	\end{align}

	Since $k\ln 2 + 1 > \ln (2^k + 1)$, we obtain,
	\begin{align}
		\label{eq:eq6}
		\Pr[\hat \mu_S \leq \mu_S - \epsilon \sqrt{\mu_{\hat S}\mu_S}] \leq \frac{\delta}{2(2^k + 1)} = \frac{\delta}{2M}
	\end{align}
	Similarly, by applying the inequality in Eq.~\ref{eq:chef_1}, we obtain the following,
	\begin{align}
		\label{eq:eq7}
		\Pr[\hat \mu_{S} \geq \mu_{S} + \epsilon \sqrt{\mu_{\hat S}\mu_{S}}] \leq \frac{\delta}{2(2^k + 1)} = \frac{\delta}{2M}
	\end{align}
	
	Combining Eq.~\ref{eq:eq6} and Eq.~\ref{eq:eq7} proves Lem.~\ref{lem:app}.
\end{proof}

Lem.~\ref{lem:app} states that if we have at least $T^*=\frac{n\Lambda}{\E[D(\hat S)]}$ \TRR{} sets then a set $S \subset V_{I}$ is approximated within an additive error of $\epsilon\sqrt{\mu_{\hat S}\mu_S}$ with probability $(1-\frac{\delta}{M})$. As a consequence, the next lemma shows that \SISI{} generates at least $T^*$ \TRR{} set, thus the approximation of $S \subseteq V_I$ in \SISI{} is also good.
\begin{Lemma}[Stopping condition]
	\label{lem:rr}
	The number of \TRR{} sets generated by \SISI{} when it stops satisfies,
	\begin{align}
		\label{eq:lem_rr}
		\Pr[|\R| \leq T^*] \leq \frac{\delta}{M}
	\end{align}
\end{Lemma}
\begin{proof}
	We also define the random variable $X_j$, samples $\{ X^1_j,X^2_j,\dots \}$ for the set $\hat S$ returned by \SISI{} similar to the proof of Lem.~\ref{lem:app}. Starting from the left-hand side of Eq.~\ref{eq:lem_rr}, we manipulate as follows,
	\begin{align}
		\label{eq:eq25}
		\Pr[|\R| \leq T^*] = \Pr[\sum_{i=1}^{|\R|}X^i_j \leq \sum_{i=1}^{T^*}X^i_j]
	\end{align}
	Since $|R_{Blue}^{-}(\hat S)| + |R_{Red}^{+}(\hat S)| =\sum_{i=1}^{|\R|}X^i_j$ and \SISI{} stops when  $|R_{Blue}^{-}(\hat S)| + |R_{Red}^{+}(\hat S)| \ge \Lambda$, Eq.~\ref{eq:eq25} is equivalent to,
	\begin{align}
		\Pr[|\R| \leq T^*] &\leq \Pr[\Lambda \leq \sum_{i=1}^{T^*}X^i_j] = \Pr[\frac{n}{T^*}\Lambda \leq \frac{n}{T^*}\sum_{i=1}^{T^*}X^i_j] \nonumber \\
		& = \Pr[\Delta\frac{n}{T^*}\Upsilon (1 + \epsilon) \leq \frac{n}{T^*}\sum_{i=1}^{T^*}X^i_j]
	\end{align}
	Recall that $T^*=\frac{n\Upsilon}{\E[D(\hat S)]}$ or $\E[D(\hat S)]=\frac{n\Upsilon}{T^*}$ and, thus,
	\begin{align}
		\label{eq:eq21}
		\Pr[|\R| \leq T^*] &\leq \Pr[\E[D(\hat S)] (1 + \epsilon) \leq \frac{n}{T^*}\sum_{i=1}^{T^*}X^i_j] \nonumber \\
		& = \Pr[\E[D(\hat S)] (1 + \epsilon) \leq D_{T^*}(\hat S)]
	\end{align}
	From Lem.~\ref{lem:app}, if we have $T^*$ RR sets, we obtain,
	\begin{align}
		\Pr[D_{\R}(S) \geq \E[D(S)] + \epsilon \sqrt{\E[D(S)] \cdot \E[D(\hat S)]}] \leq \frac{\delta}{M} \nonumber
	\end{align}
	\noindent for set $S$. Replacing $S$ by $\hat S$ gives,
	\begin{align}
		\Pr[D_{T^*}(\hat S) \geq \E[D(\hat S)] + \epsilon \E[D(\hat S)]] \leq \frac{\delta}{M} \nonumber
	\end{align}
	The left side is exactly the Eq.~\ref{eq:eq21} and thus,
	\begin{align}
		\Pr[|\R| \leq T^*] \leq \frac{\delta}{M}
	\end{align}
	That completes the proof of Lem.~\ref{lem:rr}.
\end{proof}

Based on Lem.~\ref{lem:app} and Lem.~\ref{lem:rr}, we are sufficient to prove the $\Delta\frac{2}{(1 - \epsilon)^2}$-approximation factor of \SISI{}.
\begin{theorem}
	\label{theo:approx}
	Let $OPT = \E[D(S^*)]$ be the optimal value of $\E[D(S)]$ at $S^*$. \SISI{} returns an $\Delta\frac{2}{(1 - \epsilon)^2}$-approximate solution $\hat S$ with probability of at least $(1-\delta)$ or,	
	\begin{align}
		\Pr[\E[D(\hat S)] \leq \Delta\frac{2}{(1 - \epsilon)^2}OPT] \geq 1 - \delta
	\end{align}
\end{theorem}
\begin{proof}
	From Lem.~\ref{lem:app}, we obtain,
	\begin{align}
		\Pr[|D_{\R}(S) - \E[D(S)]| \geq \epsilon \sqrt{\E[D(S)] \cdot \E[D(\hat S)]}] \leq \frac{\delta}{M} \nonumber
	\end{align}
	for a particular subset $S \subseteq V_{I}$ if there are at least $T^*$ \TRR{} sets. Furthermore, Lem.~\ref{lem:rr} states that \SISI{} generates at least $T^*$ \TRR{} sets with probability at least $\frac{\delta}{M}$.
	Taking union bound over all subsets $S \subseteq V_{I}$ (note that there are $2^k$ such subsets) to the above probability and the probability of \SISI{} generating at least $T^*$ \TRR{} sets in Lem.~\ref{lem:rr}, we achieve,
	\begin{align}
		\Pr[|D_{\R}(S) - \E[D(S)]| \geq \epsilon \sqrt{\E[D(S)] \cdot \E[D(\hat S)]}] \leq \delta\nonumber
	\end{align}
	for every set $S$. Thus, both
	\begin{align}
		\label{eq:eq22}
		D_{\R}(\hat S) \geq \E[D(\hat S)] - \epsilon \E[D(\hat S)]
	\end{align}
	and
	\begin{align}
		\label{eq:eq23}
		D_{\R}(S^*) \leq OPT + \epsilon \sqrt{OPT \cdot \E[D(\hat S)]}
	\end{align}
	happen with probability at least $(1-\delta)$. Plugging $D_{\R}(\hat S) \leq \Delta D_{\R}(S^*)$ achieved by submodular-cost covering to Eq.~\ref{eq:eq23},
	\begin{align}
		D_{\R}(\hat S) \leq \Delta (OPT + \epsilon \sqrt{OPT \cdot \E[D(\hat S)]})
	\end{align}
	then combining with Eq.~\ref{eq:eq22} gives,
	\begin{align}
		\E[D(\hat S)] - \epsilon \E[D(\hat S)] \leq \Delta (OPT + \epsilon \sqrt{OPT \cdot \E[D(\hat S)]}) \nonumber
	\end{align}
	or
	\begin{align}
		\frac{\E[D(\hat S)]}{OPT} \leq \frac{\Delta}{1 - \epsilon - \epsilon \Delta \sqrt{\frac{OPT}{\E[D(\hat S)]}}}
	\end{align}
	This inequality is valid only when $1-\epsilon-\epsilon \Delta \sqrt{\frac{OPT}{\E[D(\hat S)]}} > 0$ which means $\epsilon < 1/(1+\Delta\sqrt{\frac{OPT}{\E[D(\hat S)]}})$. Since $\epsilon$ is a free parameter, we can choose $\epsilon \leq 1/(1+\Delta)$ and satisfy the condition. By considering $\sqrt{\frac{\E[D(\hat S)]}{OPT}}$ as a variable and solve the quadratic inequality with $\epsilon \leq 1/(1+\Delta)$, we obtain,
	\vspace{-0.2in}
	
	\begin{align}
		\frac{\E[D(\hat S)]}{OPT} \leq \Delta\frac{2}{(1 - \epsilon)^2}
	\end{align}
	which states the $\Delta\frac{2}{(1 - \epsilon)^2}$ approximation factor of \SISI{} and happens with probability at least $(1-\delta)$.
\end{proof}
\begin{figure*}[!ht]
	\subfloat[$|V_I| = 100$]{
		\includegraphics[width=0.32\linewidth]{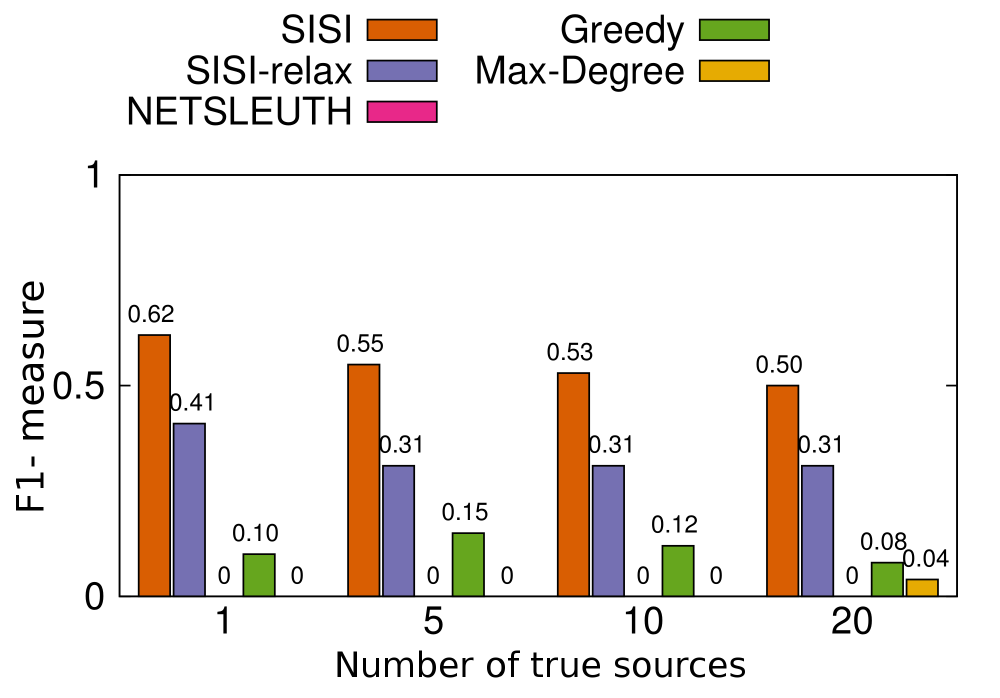}
	}
	\subfloat[$|V_I| = 500$]{
		\includegraphics[width=0.32\linewidth]{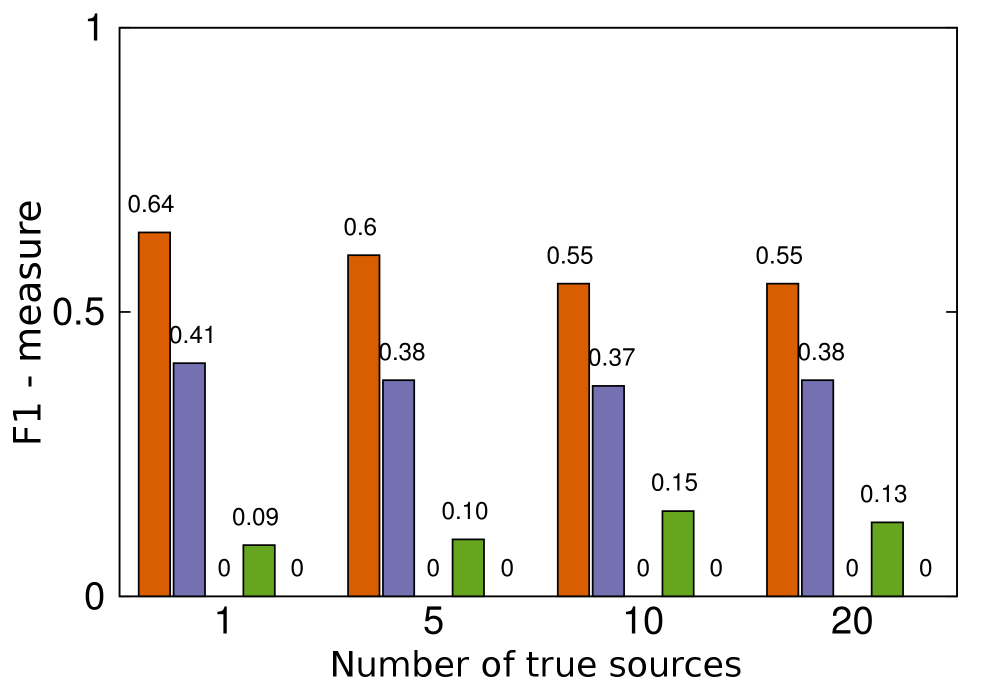}
	}
	\subfloat[$|V_I| = 1000$]{
		\includegraphics[width=0.32\linewidth]{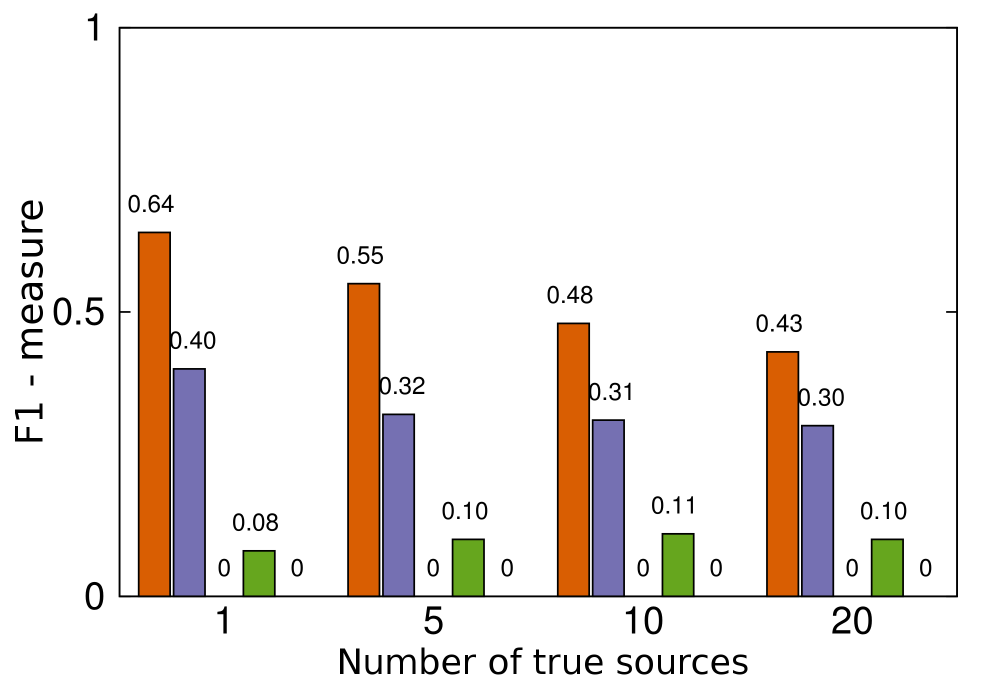}
	}
	\caption{F1-measure scores of different algorithms. Higher is better.}
	\label{fig:sol_per}
\end{figure*}
\begin{figure*}[!ht]
	\subfloat[\#sources = 1]{
		\includegraphics[width=0.24\linewidth]{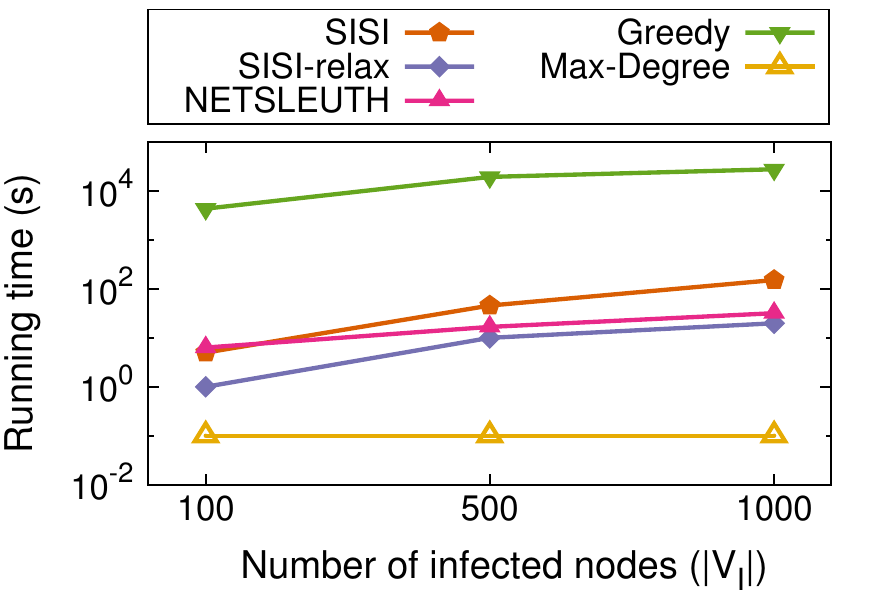}
	}
	\subfloat[\#sources = 5]{
		\includegraphics[width=0.24\linewidth]{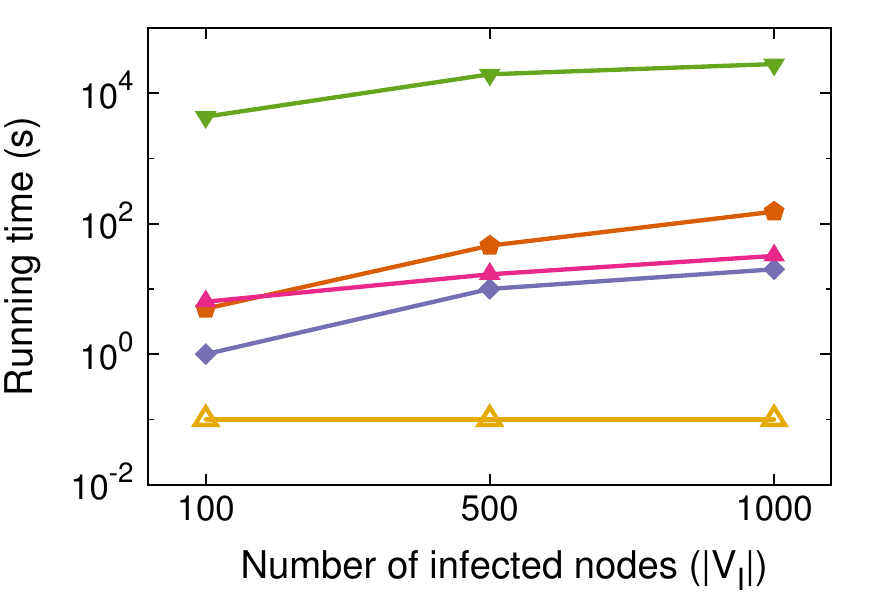}
	}
	\subfloat[\#sources = 10]{
		\includegraphics[width=0.24\linewidth]{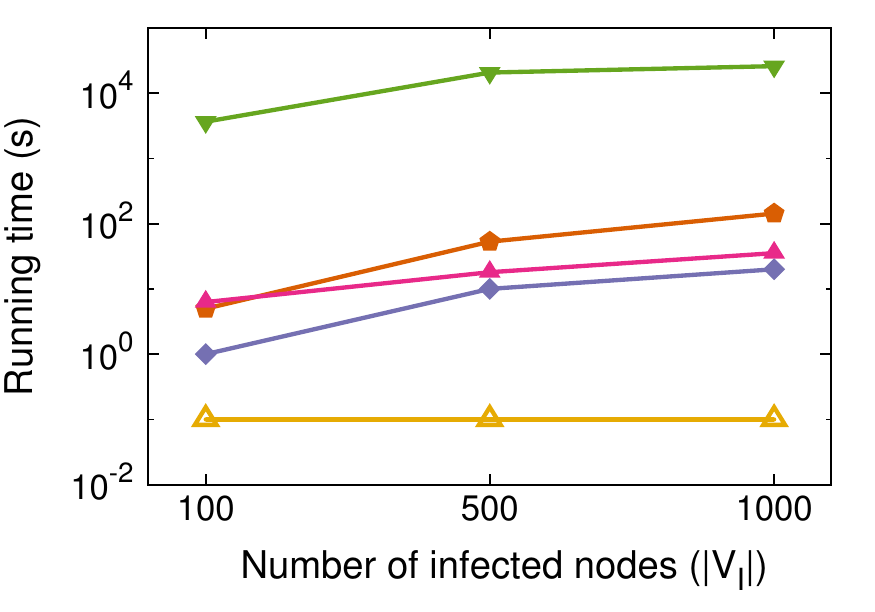}
	}
	\subfloat[\#sources = 20]{
		\includegraphics[width=0.24\linewidth]{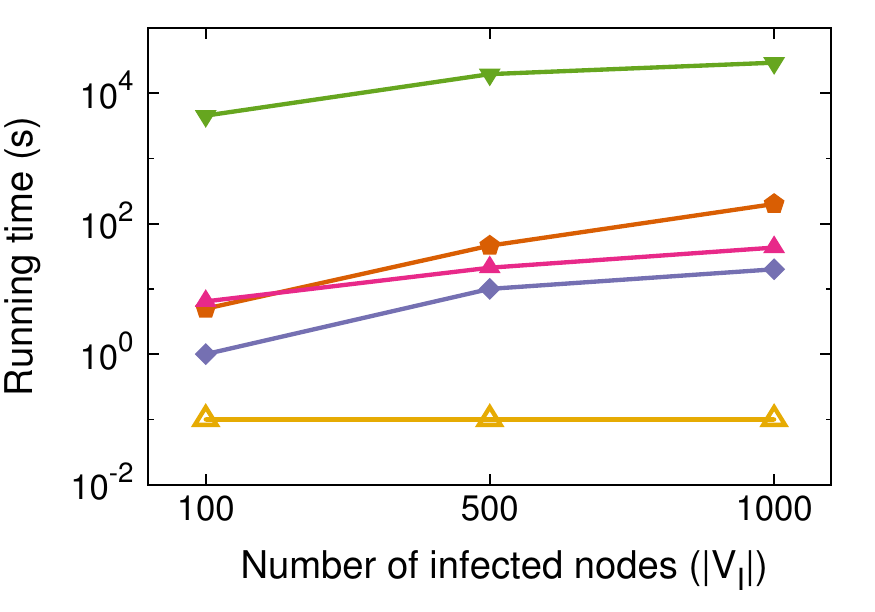}
	}
	\caption{Runtime of the tested algorithms}
	\label{fig:time}
\end{figure*}
\subsection{Time Complexity}
\label{subsec:complex}
This subsection analyzes the time complexity of \SISI{}. We analyze major procedures of the algorithm: 1) submodular-cost covering algorithm and 2) generating \TRR{} sets.
\vspace{-0.02in}
\subsubsection{Submodular-cost covering algorithm}
 Recall that the total sizes of the generated \RR{} sets is $\Lambda$ on the average.  Since the algorithm for solving the procedure to solve submodular-cost covering problem keeps doubling the number of \TRR{} sets after each round,the total complexity of this procedure is bounded loosely by $O(\Lambda^2)$. 

\subsubsection{Generating RR sets}
To determine the time complexity of generating \TRR{} sets in \SISI{}, we need analyze the time spent for generating a single \TRR{} set (Lemma~\ref{lem:ris_com}) and the expected number of \TRR{} sets. Then multiplying two numbers to get the expected total complexity. The following lemma states the complexity results with the proof in our extended version \cite{Nguyen162}.

\vspace{-0.05in}
\begin{Lemma}
	\label{lem:com_rr}
	Let $E_s$ be the set of edges connecting nodes in $V_{I}$ to nodes in $\bar V_{I}$, the complexity of generating \TRR{} sets in \SISI{} is $O(m\Lambda\Delta/|E_s|)$
\end{Lemma}
\begin{proof}
	\textit{Generating a RR Set}. As analyzed in Sec.~\ref{sec:submodular}, the expected complexity of generating a single RR set is as follows,
	\begin{align}
		\label{eq:eq12}
		C'(R_j) = \frac{\Delta m}{n} + \Delta\log (\Delta)\log (1 + \frac{\Delta m}{n^2}) \approx \frac{\Delta m}{n}
	\end{align}
	
	\textit{Number of RR set generated}. We will find an upper-bound for the number of \RR{} sets generated by \SISI{}.
	Using Wald's equation \cite{Wald47}, and that $\E[|\R|] < \infty$ we have
	\begin{align}
		\E[|\R|] \mu_{\hat S} = \Lambda
	\end{align}
	\vspace{-0.2in}
	
	\noindent Thus,
	\vspace{-0.15in}
	\begin{align}
		\label{eq:eq10}
		\E[|\R|] = \frac{\Lambda}{\mu_{\hat S}} = \frac{\Lambda n}{\E[D(\hat S,\tau,V_I)]}
	\end{align}
	
	Let $E_{s}$ be the set of edges connecting nodes in $V_I$ to nodes in $\bar V_I$, then we have
	\begin{align}
		\E[D(\hat S,\tau,V_I)] &\geq \sum_{(u,v) \in E_s}\big [(1-\Pr[\hat S,u]) + \Pr[\hat S,v]\big ]
	\end{align}
	where $(1-\Pr[\hat S,u])$ is the probability that $u \in V_I$ is not infected and $\Pr[\hat S,v]$ is the probability that $v \in V_{\bar I}$ is infected. Since $v$ is uninfected and connected with $u$, if $u$ is infected by $\hat S$, then the probability that $v$ gets the infection from $u$ is $\Pr[\hat S, v] = \beta\Pr[\hat S, u]$. Taking into the probability that $u$ is infected at least 1 step before $\tau$, we obtain $\Pr[\hat S, v] \geq \beta\Pr[\hat S, u]/(1-\beta)$ due to the binomial distribution of successes up to $\tau$ and $\tau - 1$. Thus,
	\vspace{-0.2in}
	
	\begin{align}
		& \E[D(\hat S,\tau,V_I)] \geq \sum_{(u,v) \in E_s}(1-\Pr[\hat S,u] + \frac{\beta}{1-\beta}\Pr[\hat S,u]) \nonumber \\
		& = |E_s| - (1-\frac{\beta}{1-\beta})\sum_{(u,v) \in E_s}\Pr[\hat S,u] \geq \frac{\beta}{1-\beta} |E_s|
	\end{align}
	Combining this result with Eq.~\ref{eq:eq10}, we obtain,
	\begin{align}
		\label{eq:eq11}
		\E[|\R|] \leq \frac{(1-\beta) \Lambda n}{\beta |E_s|}
	\end{align}
	
	From Eq.~\ref{eq:eq12} and Eq.~\ref{eq:eq11}, we obtains the complexity of generating \RR{} sets.
\end{proof}

Therefore, the overall complexity of \SISI{} is followed by the subsequent theorem.
\begin{theorem}
	\label{theo:com}
	Let $E_s$ be the set of edges connecting nodes in $V_{I}$ to $\bar V_{I}$, \SISI{} has $O(m\Delta\Lambda/|E_s| + \Lambda^2)$ time complexity.
\end{theorem}

From Theo.~\ref{theo:com}, we see that the complexity depends on the number of connections from infected set to the outside world $|E_s|$. That is if there are many infected nodes connected to uninfected nodes, it is easier for \SISI{} to find the sources and vice versus, if only few such connections, \SISI{} requires more time. 
\vspace{-0.05in}

\section{Experiments}
\label{sec:exp}
In this section, we  study the empirical performance of \SISI{} and compare it with the current state-of-the-art methods under the popular \SI{} and \IC{} infection models. We show that \SISI{} outperform the others in terms of detection quality, revealing major of the infection sources. In contrast, the other methods rarely find any true source of the infection. 
\setlength\tabcolsep{5pt}
\def\arraystretch{1}
\begin{table*}[!ht]\centering
	\begin{tabular}{cl cccc cccc cccc}
		\toprule
		\multicolumn{2}{c}{\textbf{$|V_I|$}} & \multicolumn{4}{c}{$100$} & \multicolumn{4}{c}{$500$} & \multicolumn{4}{c}{$1000$}\\
		\cmidrule(r){3-6}\cmidrule(r){7-10}\cmidrule(r){11-14}
		\multicolumn{2}{c}{\textit{\#sources}} & $1$ & $5$ & $10$ & $20$ & $1$ & $5$ & $10$ & $20$ & $1$ & $5$ & $10$ & $20$\\
		\midrule
		\multirow{6}{0.1\textwidth}{\textbf{Symmetric Difference} (smaller is better)} & {\color{red}\textbf{Ground-truth}} & {\color{red}\textbf{205}} & {\color{red}\textbf{173}} & {\color{red}\textbf{156}} & {\color{red}\textbf{134}} & {\color{red}\textbf{1006}} & {\color{red}\textbf{945}} & {\color{red}\textbf{938}} & {\color{red}\textbf{767}} & {\color{red}\textbf{2026}} & {\color{red}\textbf{1835}} & {\color{red}\textbf{1945}} & {\color{red}\textbf{1520}} \\
		& \SISI{} & \textbf{211} & \textbf{181} & \textbf{168} & \textbf{142} & \textbf{1013} & \textbf{962} & \textbf{971} & \textbf{792} & \textbf{2049} & \textbf{1873} & \textbf{1959} & \textbf{1541} \\
		& \SISI{}\textsf{-relax} & 246 & 215 & 218 & 202 & 1141 & 993 & 1084 & 854 & 2179 & 1903 & 2012 & 1696\\
		& \NETSLEUTH{} & 294 & 273 & 280 & 247 & 1258 & 1147 & 1193 & 971 & 2297 & 2095 & 2248 & 1751 \\
		& \textsf{Greedy} & 261 & 226 & 231 & 219 & 1152 & 1015 & 1067 & 914 & 2218 & 2214 & 2124 & 1707\\
		& \textsf{Max-Degree} & 281 & 325 & 418 & 387 & 1195 & 1091 & 1206 & 1105 & 2221 & 2167 & 2182 & 1876\\
		\bottomrule
		\multirow{6}{0.1\textwidth}{\textbf{Jaccard Distance} (larger is better)} &
		{\color{red}\textbf{Ground-truth}} & {\color{red}\textbf{1}} & {\color{red}\textbf{1}} & {\color{red}\textbf{1}} & {\color{red}\textbf{1}} & {\color{red}\textbf{1}} & {\color{red}\textbf{1}} & {\color{red}\textbf{1}} & {\color{red}\textbf{1}} & {\color{red}\textbf{1}} & {\color{red}\textbf{1}} & {\color{red}\textbf{1}} & {\color{red}\textbf{1}} \\
		& \SISI{} & \textbf{0.92} & \textbf{0.98} & \textbf{0.96} & \textbf{0.97} & \textbf{0.99} & \textbf{0.95} & \textbf{0.82} & \textbf{0.94} & \textbf{0.96} & \textbf{0.97} & \textbf{0.97} & \textbf{0.95} \\
		& \SISI{}\textsf{-relax} & 0.76 & 0.72 & 0.65 & 0.71 & 0.81 & 0.79 & 0.72 & 0.89 & 0.86 & 0.68 & 0.72 & 0.73\\
		& \NETSLEUTH{} & 0.21 & 0.24 & 0.31 & 0.37 & 0.16 & 0.29 & 0.26 & 0.41 & 0.20 & 0.17 & 0.18 & 0.21 \\
		& \textsf{Greedy} & 0.32 & 0.19 & 0.39 & 0.35 & 0.26 & 0.28 & 0.34 & 0.37 & 0.22 & 0.26 & 0.21 & 0.19\\
		& \textsf{Max-Degree} & 0.32 & 0.35 & 0.24 & 0.29 & 0.24 & 0.27 & 0.26 & 0.18 & 0.14 & 0.16 & 0.17 & 0.12\\
		\bottomrule
	\end{tabular}
	\caption{Comparison on Symmetric Difference and Jaccard-based Distance of different methods.}
	\label{tab:sym_jac}
\end{table*}
\subsection{Experimental Settings}
\subsubsection{Algorithms compared} 
Under the \SI{} model, we compare three groups of methods:
\begin{itemize}
	\item \SISI{}, a relaxed version of \SISI{}, termed \SISI\textsf{-relax}, in which we relax the approximation guarantee of \SISI{} by replacing $(k\ln2)$ in $\Upsilon$ by a smaller constant $\ln(2\times k)$ and the natural naive \textsf{Greedy} algorithm which iteratively selects one node at a time that commits the largest marginal decrease of symmetric difference. The purpose of designing \SISI{}\textsf{-relax} is to test the empirical performance changes if we have fewer \TRR{} sets.
	\item \NETSLEUTH{} \cite{Prakash12} which is the existing best algorithm in general graphs however it fails to provide any guarantee on solution quality.
	\item \textsf{Max-Degree} based method which ranks node degrees and iteratively selects nodes with highest degree until increasing the symmetric difference as the solution.
\end{itemize}

Under the \IC{} model, we compare \SISI{} with \textsf{k-effector} \cite{Lappas10} and the naive \textsf{Max-Degree} algorithm on \IC{} model.

For \SISI{} and \SISI\textsf{-relax}, we set the parameters $\epsilon=0.1, \delta=0.01$. For \textsf{k-effector}, $k$ is set to the number of true sources.
\subsubsection{Quality measures}
To evaluate the solution quality, we adopt three measures:
\begin{itemize}
	\item Symmetric difference ($\E[D(S,\tau,V_I)]$) which is separately calculated with high accuracy ($\epsilon = 0.01, \delta = 0.001$) through generating random \RR{} sets as in Subsection \ref{subsec:ris}.
	\item Jaccard distance based $Q_{JD}$ \cite{Prakash12}:
	\begin{align}
		Q_{JD}(S) = \frac{\E[JD_S(V_I)]}{\E[JD_{S^*}(V_I)]}
	\end{align}
	where $\E[JD_S(V_I)]$ is the average Jaccard distance of $S$ w.r.t. $V_I$ and computed by generating many (10000 in our experiments) infection simulations from $S$ and averaging over the Jaccard similarities between the infected sets and $V_I$. $S^*$ contains the true sources.
	\item F1-measure:
	\begin{align}
	PR(S) = \frac{|S \cap \{\text{true sources}\}|}{2|S|} + \frac{|S \cap \{\text{true sources}\}|}{2|\{\text{true sources}\}|} \nonumber
	\end{align}
	
	This accurately captures our ultimate goal of \ISI{} problem: finding both the true sources and the correct number of sources. We also define true source detection rate (\%) as $100\frac{|S \cap \{\text{true sources}\}|}{|\{\text{true sources}\}|}$.
\end{itemize}
Both $Q_{JD}(S)$ and $PR(S)$ are ranging in $[0,1]$ and larger is better. $\E[D(S,\tau,V_I)]$ is nonnegative and smaller is better.
\subsubsection{Datasets} For experimental purposes, we select a moderate-size real network - NetHEPT with 15233 nodes and 62796 edges that is actually the largest dataset ever tested on \ISI{} problem. We comprehensively carry experiments on NetHEPT with various numbers of sources $\{1, 5, 10, 20\}$, chosen uniformly random, and the propagation time $\tau$ is chosen so that the infection sizes reach (or exceed) predefined values in the set $\{100, 500, 1000\}$. For each pair of the two values, we generated 10 random test cases with $\beta = 0.05$ and then ran each method on these random tests and took the average of each quality measure over 10 such results.

\subsubsection{Testing Environments} We implement \SISI{}, \SISI\textsf{-relax}, \textsf{Greedy} and \textsf{Max-Degree} methods in C++, \NETSLEUTH{} is in Matlab code and obtained from the authors of \cite{Prakash12}. We experiment on a Linux machine with an 8 core 2.2 GHz CPU and 100GB RAM.
\vspace{-0.05in}
\subsection{Experiments on real network and SI model}
\textbf{Comparing solution quality.} The solution quality measured are the  true infection sources discovery rate, symmetric difference (our objective) and Jaccard-based distance \cite{Prakash12}

\textit{True source discovery.} Fig.~\ref{fig:sol_per} reports the F1-measure scores of the tested algorithms. Note that this score has not been used in previous works \cite{Prakash12,Luo13} since previous methods can only find nodes that are within few hops from the sources, but not the sources themselves. 
As shown in the figure, \SISI{} and \SISI\textsf{-relax} have the best performance. More than 50\% of the true sources was discovered by \SISI{} and 35\% by \SISI\textsf{-relax} that exquisitely surpass \NETSLEUTH{}, \textsf{Max-Degree} with 0\% and \textsf{Greedy} with roughly 10\%.

\setlength\tabcolsep{3pt}
\def\arraystretch{1}
\begin{table}[!ht]\centering
	\begin{tabular}{c ccccc}
		\toprule
		\#src & \multicolumn{1}{c}{\SISI{}} & \multicolumn{1}{c}{\SISI{}\textsf{\textsf{-relax}}} & \multicolumn{1}{c}{\textsf{NETS.}} & \multicolumn{1}{c}{\textsf{Greedy}} & \multicolumn{1}{c}{\textsf{Max-Degree}}\\
		\midrule
		1 & \textbf{91.4} & 84.2 & 0 & 14.5 & 0\\
		5 & \textbf{79.7} & 53.9 & 0 & 15.2 & 0\\
		10 & \textbf{74.1} & 52.3 & 0 & 11.8 & 0\\
		20 & \textbf{77.3} & 56.5 & 0 & 9.6 & 0\\
		\bottomrule
	\end{tabular}
	\caption{True sources detected (\%) with $|V_I| = 1000$.}
	\label{tab:true_pos}
\end{table}
We also present the true source detected rates of different methods in Tab.~\ref{tab:true_pos} since this is an important aspect (positive rate) of \ISI{} problem. The table shows accurate detection of \SISI{} and \SISI{}\textsf{-relax}. More than 70\% and 50\% of true sources are identified by \SISI{} and \SISI{}\textsf{-relax} respectively while \NETSLEUTH{} and \textsf{Max-Degree} cannot detect any source.

\textit{Symmetric difference.} Tab.~\ref{tab:sym_jac} shows the $\E[D(S,\tau,V_I)]$ values where $S$ is the returned solution of each algorithm with various number of true sources and sizes of infection cascades. In all the cases, \SISI{} largely outperforms the other methods and obtains very close values to the true sources. The superiority of \SISI{} against the \SISI\textsf{-relax} and \textsf{Greedy} confirms the good solution guarantee of \SISI{}. \NETSLEUTH{} and \textsf{Max-Degree} optimize different criteria, i.e., description length (MDL) and node degree, and thus show poor performance in terms of symmetric difference. \SISI\textsf{-relax} is consistently the second best method and preserves very well the performance of \SISI{}. 

\textit{Jaccard distance.} We use $Q_{JD}(S)$ as in \cite{Prakash12} to evaluate the algorithms and plot the results in Tab.~\ref{tab:sym_jac}. In this case, the closer value of $Q_{JD}(S)$ to 1 indicates better solution. In terms of $Q_{JD}(S)$, we observe the similar phenomena as measured by symmetric difference that \SISI{} achieve drastically better solution than the others and the results of \SISI\textsf{-relax} approach those of \SISI{} very well with much fewer \TRR{} sets.

\textbf{Comparing running time.}
Fig.~\ref{fig:time} illustrates the running time of the algorithms in the previous experiments. We see that \SISI{} is slower than \NETSLEUTH{} and \SISI{}\textsf{-relax} but the differences are minor while it provides by far better accuracy than other algorithms. \SISI{}\textsf{-relax} obtains possibly the best balance among all: faster than \NETSLEUTH{} and providing good solution quality as shown previously.
\subsection{Experiments on the IC model}
\textbf{Set up.} We compare \SISI{} with the dynamic programming algorithm, temporarily called \textsf{k-effector}, in \cite{Lappas10} when the infection process follows the IC model. Similar to other experiments, we simulate the infection process under the IC model with 4 different numbers of sources, i.e., 1, 5, 10, 20 and run \SISI{} and \textsf{k-effector} on the resulting cascades. For each setting, we carry 10 simulations and report the average results. 
Note that the solution for \textsf{k-effector} in \cite{Lappas10} requires the number of sources as an additional input parameter and for simplicity, we provide the true number of sources used in the simulation processes. \SISI{}, however, do not require this information. We report the results in Table~\ref{tab:ic_model}.
\setlength\tabcolsep{2pt}
\def\arraystretch{1}
\begin{table}[!ht]\centering
 	\begin{tabular}{c ccc ccc}
 		\toprule
	 	\multirow{2}{0.04\textwidth}{\#src} & \multicolumn{3}{c}{Symmetric Difference} & \multicolumn{3}{c}{F1-measure} \\
	 	\cmidrule(r){2-4}\cmidrule(r){5-7}
	 	& \multicolumn{1}{c}{\SISI{}} & \multicolumn{1}{c}{\textsf{k-effector}} & \multicolumn{1}{c}{\textsf{Max-Deg.}} & \multicolumn{1}{c}{\SISI{}} & \multicolumn{1}{c}{\textsf{k-effector}} & \textsf{Max-Deg.} \\
	 	\midrule
	 	1 & \textbf{6.6} & 18.4 & 42.3 & \textbf{0.57} & 0 & 0.02\\
	 	5 & \textbf{55.1} & 103.4 & 176.9 & \textbf{0.53} & 0.02 & 0\\
	 	10 & \textbf{25.2} & 72.6 & 154.1& \textbf{0.49} & 0.03 & 0\\
	 	20 & \textbf{203.7} & 295.2 & 384.7& \textbf{0.52} & 0.05 & 0.03\\
 		\bottomrule
 	\end{tabular}
 	\caption{Comparison under the \IC{} model.}
 	\label{tab:ic_model}
\end{table}

 
\textbf{Results.} It is clear from Table~\ref{tab:ic_model} that \SISI{} massively outperforms \textsf{k-effector} in terms of both symmetric difference and true source recovering ability. In summary, for any value of the number of true sources $k$, \SISI{} always returns solution with symmetric difference equal half of the one returned by \textsf{k-effector}. In terms of true source discovery ability, while \textsf{k-effector} almost detects none of the true sources, \SISI{} consistently achieves the F1-measure of at least 50\%.


\section{Discussion and Conclusion}
\label{sec:conclusion}
We present \SISI{} the first approximation algorithm for multiple source detection in general graphs which also works very well in practice. The algorithm can be extended to several other diffusion models and settings with little modification  on the sampling procedure as outlined below.
%

\emph{Incomplete Observation \cite{Farajtabar15, Karamchandani13}.} In many cases,  we can only observe the states (infected/not infected) for a subset $O \subsetneq V$ of nodes in the network. In those cases, we need to modify the Fast \TRIS{} sampling Algorithm in Line 1 and pick a node $u$ uniformly in $O$ (instead of $V$) and allow the sources to be from $V_I$ or unknown state nodes.

However, the \SISI{} cannot be directly adapted to non-progessive models in which a node can switch from an infected state into uninfected state. Thus approximation algorithm for source detection in non-progressive models leaves an open question and is among our future work.
\vspace{-0.05in}

\bibliographystyle{abbrv} 
\bibliography{infection,dijkstra,nphard,social,pids,isi}

\balance
\end{document}